\documentclass[aps,prl,superscriptaddress,10pt,article,showpacs,longbibliography,twocolumn]{revtex4-2}

\usepackage{blindtext}
\usepackage{lipsum}
\usepackage{graphics}
\usepackage{amsmath}
\usepackage{graphicx}
\usepackage{graphics}
\usepackage{amssymb}
\usepackage{pifont}
\usepackage{verbatim}
\usepackage{physics}
\usepackage{float}
\usepackage[normalem]{ulem}
\usepackage[dvipsnames]{xcolor}
\usepackage{bbm}
\usepackage{enumitem} 
\usepackage{multirow}

\usepackage{dsfont}
\newcommand{\identity}{\mathds{1}}

\usepackage{graphicx}
\usepackage{bm}
\usepackage{amsmath}
\usepackage{physics}
\usepackage{amssymb}
\usepackage{amsfonts}
\usepackage{amsthm}
\usepackage{bbm}
\usepackage{mathtools}
\usepackage{braket}
\usepackage[normalem]{ulem}
\usepackage{wrapfig}
\usepackage{tikz}
\usepackage{dsfont}
\usepackage{comment}
\usepackage{thmtools,thm-restate}
\usepackage[export]{adjustbox}

\definecolor{blueviolet}{rgb}{0.2, 0.2, 0.6}
\definecolor{webgreen}{rgb}{0,.5,0}
\definecolor{webbrown}{rgb}{.6,0,0}
\usepackage[bookmarks=false,
	colorlinks=true, %allcolors=blueviolet,
	urlcolor=webbrown, 
	linkcolor=blueviolet, 
	citecolor=webgreen,
	pdfstartpage=1,
	pdfstartview={FitH},  % FitBH
	bookmarksopen=false
	]{hyperref}
    
\newtheorem{theorem}{Theorem}

\newtheorem*{theorem*}{Theorem}

\newtheorem*{task*}{Task}
\newtheorem*{proposition*}{Proposition}

\newcommand{\bs}{\boldsymbol}
\newcommand{\be}{\begin{equation}}
\newcommand{\ee}{\end{equation}}

\DeclareMathOperator*{\E}{{\mathbb{E}}}

\begin{document}

\title{Probing the classical complexity of quantum dynamics experiments}

\author{Thomas Schuster}
\affiliation{Walter Burke Institute for Theoretical Physics and Institute for Quantum Information and Matter, California Institute of Technology, Pasadena, California 91125, USA}
\affiliation{Google Quantum AI, Venice, California 90291, USA}

\author{Andreas Elben}
\affiliation{PSI Center for Scientific Computing, Theory and Data  and ETHZ-PSI Quantum Computing Hub,
Paul Scherrer Institute, CH-5232 Villigen-PSI, Switzerland}

\begin{abstract}
A confluence of recent works has shown that many quantum circuits and  dynamics are efficiently simulable by classical algorithms that track local information, even when conventional complexity measures such as the entanglement and magic are high.
Here, we introduce a novel measure of complexity, the \emph{reactivity}, to capture this new method of classical attack.
Unlike conventional complexity measures, the reactivity does not capture a property of a quantum state or operator in isolation, but rather a quantum experiment as a whole.
We provide numerical and rigorous evidence that quantum experiments with low reactivity are simple by a host of measures: they are efficient to classically simulate, learn, and fast-forward.
This motivates the search for quantum experiments with high reactivity, which may evade these simplistic features.
To this end, we introduce easily implementable experimental protocols---dubbed \emph{Pauli path spectroscopy}---that allow one to efficiently measure the reactivity of any quantum experiment of interest.
Our protocols are applicable even when the experiment itself is beyond the reach of  classical simulation.
\end{abstract}

\maketitle

What is the classical complexity of  simulating quantum dynamics?
A colloquial wisdom states that classical simulation of large quantum systems is hard, and hence quantum simulation is a fail-proof application of quantum computers~\cite{feynman2018simulating}.
Nevertheless, our understanding and characterization of quantum advantage in physical settings, outside of sampling tasks~\cite{aaronson2011computational,arute2019quantum,bouland2019complexity,movassagh2023hardness,bouland2022noise}, is relatively limited.
This is especially felt  in current quantum experiments, which often realize physical dynamics whose classical complexity is entirely unknown beforehand.

Adding to this puzzle, a recent confluence of classical algorithms has suggested that many quantum experiments can be simulated far more efficiently than previously believed~\cite{white2018quantum,rakovszky2022dissipation,ye2020emergent,holdendye2026diameter,guharoy2026reweighted,klein2022time,artiaco2024efficient,kuprov2007polynomially,karabanov2011accuracy,surace2019simulating,bhateja2026compressed,white2023effective,beguvsic2023fast,rudolph2023classical,angrisani2024classically,yithomas2024comparing,begusic2025real,loizeau2025quantum,rudolph2025pauli,fontana2023classical,dowling2026noise,shao2025characterizing,xu2026classical,gao2018efficient,aharonov2023polynomial,schuster2025polynomial,gonzalez2025pauli,angrisani2026simulating,cirstoiu2024fourier,martinez2025efficient,mele2024noise,granet2025dilution}.
For several decades, our leading metrics for understanding the classical hardness of quantum dynamics were the entanglement and magic.
These quantify the hardness of simulations via tensor networks \cite{paeckel2019time,cirac2021matrix} and stabilizer methods \cite{bravyi2019simulation}, respectively.
Recent algorithms succeed via a starkly different approach: by retaining \emph{local} information about a time-evolved state or operator, but purposefully forgetting complex \emph{non-local} information.
This builds on a physical intuition, that the evolution of local properties of many-body systems should be explainable by local information, even when complex non-local features are present.
In practice, such algorithms often store time-evolved observables in the Pauli basis~\cite{white2023effective,beguvsic2023fast,rudolph2023classical,angrisani2024classically,yithomas2024comparing,begusic2025real,rudolph2025pauli,fontana2023classical,loizeau2025quantum,shao2025characterizing,dowling2026noise,xu2026classical,gao2018efficient,aharonov2023polynomial,schuster2025polynomial,gonzalez2025pauli,angrisani2026simulating,cirstoiu2024fourier,martinez2025efficient}, or perform tensor network truncations that preference local information~\cite{white2018quantum,rakovszky2022dissipation,ye2020emergent,holdendye2026diameter}.
These approaches have enabled efficient and accurate classical simulations of quantum circuits and dynamics that appear extremely complex by traditional  measures~\cite{white2018quantum,rakovszky2022dissipation,ye2020emergent,holdendye2026diameter,guharoy2026reweighted,klein2022time,artiaco2024efficient,kuprov2007polynomially,karabanov2011accuracy,surace2019simulating,bhateja2026compressed,white2023effective,beguvsic2023fast,rudolph2023classical,angrisani2024classically,yithomas2024comparing,begusic2025real,loizeau2025quantum,rudolph2025pauli,fontana2023classical,dowling2026noise,shao2025characterizing,xu2026classical,gao2018efficient,aharonov2023polynomial,schuster2025polynomial,gonzalez2025pauli,angrisani2026simulating,cirstoiu2024fourier,martinez2025efficient,mele2024noise,granet2025dilution}.

This progress raises essential questions.
Which quantum circuits and dynamics can be  simulated with local information algorithms, and which cannot? 
Can we formulate a simple metric, analogous to the entanglement and magic, to understand these methods' success?
And can quantum experiments \emph{themselves} help to answer these questions, especially in beyond-classical regimes?

\begin{figure}
\centering
\includegraphics[width=0.95\columnwidth]{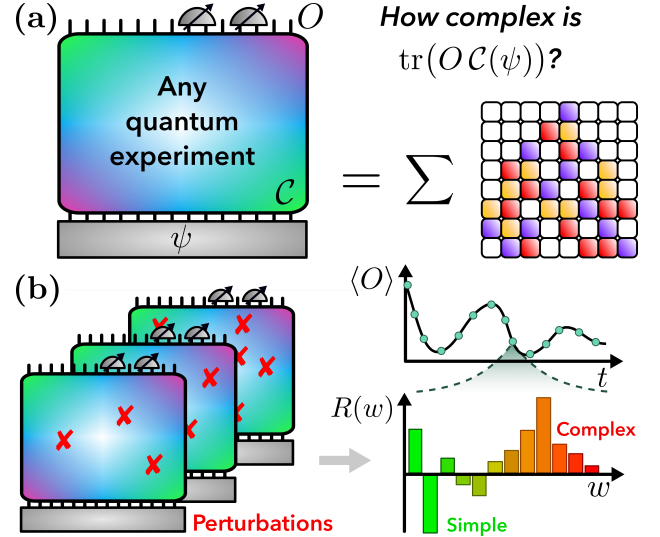}
\caption{
\textbf{(a)} 
Any quantum experiment, $\tr(O\,\mathcal{C}(\psi))$, can be expanded as a sum over \emph{Pauli paths} (illustrated as space-time patterns of Pauli operators).
\textbf{(b)} Pauli path spectroscopy probes this decomposition by measuring the response of the experiment to controlled local perturbations.
This allows one to extract the \emph{reactivity function}, $R(w)$, which quantifies the total contribution to $\tr(O\,\mathcal{C}(\psi))$ from Pauli paths of weight $w$.
Contributions from small $w$ (green) are easy to classically simulate, while contributions from high $w$ (red) are hard to simulate for classical  algorithms based on local information.}
\label{fig:1}
\end{figure}

In this work, we address these questions by introducing a physical metric---the \emph{reactivity}---to characterize the complexity of quantum experiments with respect to local information algorithms.
The reactivity measures how much a quantum experiment is influenced by local versus non-local information.
This in turn  determines how
strongly the  experiment ``reacts'' to external perturbations.
Unlike conventional complexity measures, the reactivity does not characterize a quantum state or operator in isolation, but rather a quantum experiment as a whole.
The reactivity is inspired by, and unifies and generalizes, several existing concepts: the operator backflow of infinite-temperature correlation functions~\cite{rakovszky2018diffusive,von2022operator}; the effective volume of quantum circuits~\cite{kechedzhi2023effective}; and the influence of classical boolean functions~\cite{odonnel2014analysis}.

Our main contributions are threefold.
First, building on recent Pauli path decompositions of quantum circuit dynamics~\cite{gao2018efficient,aharonov2023polynomial,schuster2025polynomial,gonzalez2025pauli}, we introduce the reactivity and describe its connection to the response of an experiment to local perturbations.
Second, we perform detailed analyses that support use of the reactivity as a complexity metric.
These include numerical evidence that the reactivity correlates with the computational resources of local information algorithms for quantum many-body dynamics, and rigorous proofs that circuits with low reactivity are efficient to  simulate classically, and to learn and fast-forward on quantum devices.
Finally, we outline how scalable experimental measurements of the reactivity, for any quantum circuit or dynamics, can be easily performed.
Our protocol measures the reaction of the circuit to purposefully-inserted perturbations, combined with straightforward classical post-processing.

We remark that, similar to other complexity metrics, the reactivity can  easily be ``spoofed'' by an adversarial experimenter.
That is, one can design circuits that appear complex from the point of view of the reactivity, but which are simple by other means.
To our knowledge, this is a shared trait of every practical complexity metric, including the entanglement and magic.
This highlights that these metrics should be used in complement with one another, and physical intuition, in practice.

\emph{Background.}---Let us begin by providing a brief review of local information (LI) algorithms, and the intuition for their success.
We use this term to encompass a broad suite of recent numerical algorithms, including those that store sparse approximations of time-evolved operators in the Pauli basis~\cite{beguvsic2023fast,begusic2025real,loizeau2025quantum,rudolph2025pauli} and those that merge related approximations with tensor network methods~\cite{white2018quantum,rakovszky2022dissipation,klein2022time}.
The key shared feature of these algorithms is that, when applied to physical systems, they tend to keep track primarily of local information.
In some cases, this is explicit in the algorithm instructions~\cite{white2018quantum,rakovszky2022dissipation}; in others, it occurs implicitly due to physical effects~\cite{beguvsic2023fast,begusic2025real,loizeau2025quantum,rudolph2025pauli}.
These algorithms have led to startlingly accurate simulations of highly-entangled quantum many-body dynamics~\cite{beguvsic2023fast,begusic2025real,loizeau2025quantum}, as well as significant interest for benchmarking quantum advantages in experiments~\cite{kechedzhi2023effective} and quantum optimization and machine learning tasks~\cite{rudolph2023classical}.

The most striking feature of LI algorithms is their ability to accurately predict \emph{expectation values}, even when \emph{exact simulation} of quantum states and operators is classically hard.
To build intuition for this separation, we can consider its most trivial example, a local random circuit.
Here, exact simulations of quantum states and operators are believed to be hard due to the hardness of sampling~\cite{arute2019quantum,bouland2019complexity,movassagh2023hardness}.
Nonetheless, expectation values are trivial to predict, as they decay exponentially to zero~\cite{schuster2024random}.
This separation arises because the complex components of time-evolved states and operators do not \emph{contribute} to the expectation value, despite their existence in the exact time-evolved operator and state.

While this example illustrates the possibility of a large distinction between computing expectation values and exact simulation, most quantum dynamics  of interest are not so trivial.
Indeed, the premier applications of LI algorithms have been to strongly-interacting Hamiltonian dynamics with non-zero expectation values up to late times.
This application is motivated by a simple  intuition: that local dynamics should be explainable by local information.
Hence, even if complex non-local information is physically present in exact time-evolution, it may be safely neglected to predict local expectation values, much as in the random circuit example.

Despite this intuition, and significant numerical studies, the precise extent to which LI algorithms succeed and fail has remained a broad and fundamental open question.
Developing the theoretical and experimental tools to answer this question is the  goal of our work.

\emph{The reactivity function.}---To capture the complexity of LI algorithms, we would like to distinguish the contributions to a given quantum experiment from local information versus non-local information.
This is naturally achieved using the Pauli path formalism~\cite{gao2018efficient,aharonov2023polynomial,schuster2025polynomial,loizeau2025quantum,rudolph2025pauli,begusic2025real,gonzalez2025pauli}.
We consider any expectation value (i.e.~quantum experiment), $\tr( O \, \mathcal{C}(\psi))$, where $\psi$ is an initial state  on $N$ qubits, $\mathcal{C} = \mathcal{C}_t \circ \ldots \circ \mathcal{C}_1$ is a quantum circuit with $t$ layers, and $O$ is an observable.
We can decompose the expectation value into a sum of contributions from \emph{paths} of Pauli operators,
\begin{equation}
    \tr(O \mathcal{C}(\psi)) = \sum_{\vec P} A_{\vec P},
\end{equation}
where each path is a sequence of Pauli operators, $\vec P = (P_0,\ldots,P_t)$, which contributes to the experiment with an amplitude, $A_{\vec P} = \overline{\tr}(O P_t) \cdot \prod_{s=1}^t \overline{\tr}(P_s \mathcal{C}_s(P_{s-1})) \cdot \tr(P_0 \psi)$ where $\overline{\tr}(\cdot) = 1/2^N \tr(\cdot)$.

The Pauli path decomposition naturally separates contributions from local information, corresponding to paths of small Pauli operators, from non-local information, corresponding to large Pauli operators.
Our metric of complexity is now straightforward.
For any time $\tau \in [0,t]$, we define the \emph{reactivity function}, $R_\tau(w)$, as the total contribution from all Pauli paths with weight $w$ at time $\tau$,
\begin{equation}
    \tr( O \, \mathcal{C}(\psi)) = \sum_w R_\tau(w), 
    \vspace{-1mm}
\end{equation}
where $R_\tau(w) = \sum_{\vec P} \delta_{w[P_\tau] = w} A_{\vec P}$, and $w[P_\tau]$ is the number of non-identity elements in $P_\tau$.
We also define the \emph{global reactivity function}, $R(w)$, in the same manner, but with respect to the space-time weight, $w[\vec P] = \sum_{\tau=0}^t w[P_\tau]$.

The reactivity function decomposes any quantum experiment into contributions sorted by their locality.
If an experiment has high reactivity, i.e.~non-negligible contributions from large weights $w$, it will likely appear complex to LI algorithms.
In particular, owing to the exponential increase in the number of Pauli operators with the weight, one would typically expect the simulation complexity to grow exponentially in the required $w$.
On the other hand, if an experiment has low reactivity, i.e.~contributions only from small weights, it will be easy to simulate with LI algorithms.
In practice, most applications of LI algorithms involve weights $w \lesssim 10$~\cite{fn1}.

We make two further remarks.
First, our name---the \emph{reactivity}---is inspired by a close connection between our complexity metric and the \emph{reaction} of a quantum experiment to local noise.
In particular, if one applies local depolarizing noise of strength $\gamma$ on each qubit at time $\tau$, the experimental outcome becomes,
\begin{equation} \label{eq:reactivity-noise}
    \tr( O \, \mathcal{C}_\gamma(\psi))  = \sum_w e^{-\gamma w} R_\tau(w).
\end{equation}
A similar relation holds for the global reactivity, for noise applied at every circuit layer.
Hence, contributions from large weights $w$ react rapidly to even small amounts of noise, $\gamma \sim 1/w$, while contributions from small weights are less affected~\cite{fn2}.
This physical connection will form the basis of our experimental protocols later on.
It also suggests that the reactivity may be useful more broadly, for example in understanding the impact of noise, or the efficacy of noise mitigation strategies~\cite{granet2025dilution,suchsland2026quantum}.

Second, we emphasize that the reactivity probes fundamentally different characteristics than conventional measures of operator growth, such as out-of-time-order correlators~\cite{xu2022scrambling} and size distributions~\cite{roberts2018operator}, which quantify the weight needed to predict expectation values in \emph{time-reversal} experiments but not elsewhere~\cite{schuster2022operator,cotler2023information,schuster2022learning,schuster2024random,schuster2025strong,abanin2025constructive}.
The reactivity generalizes the operator backflow~\cite{rakovszky2018diffusive} from infinite-temperature correlators to any quantum experiment.
The global reactivity is related to the effective quantum volume~\cite{kechedzhi2023effective}, which equals its first moment.

\begin{figure*}
  \centering
  \includegraphics[width=\textwidth]{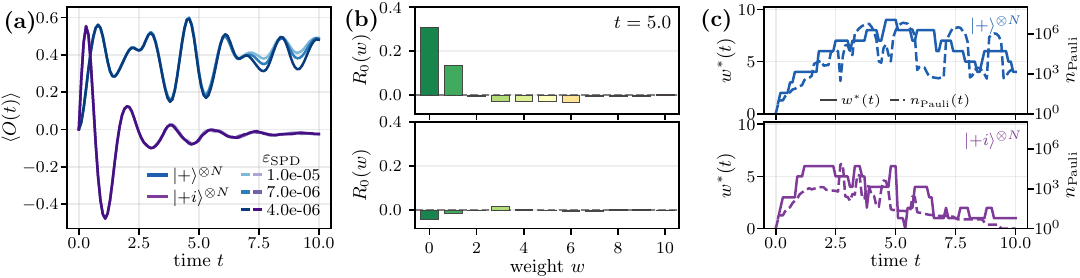}  \includegraphics[width=\textwidth]{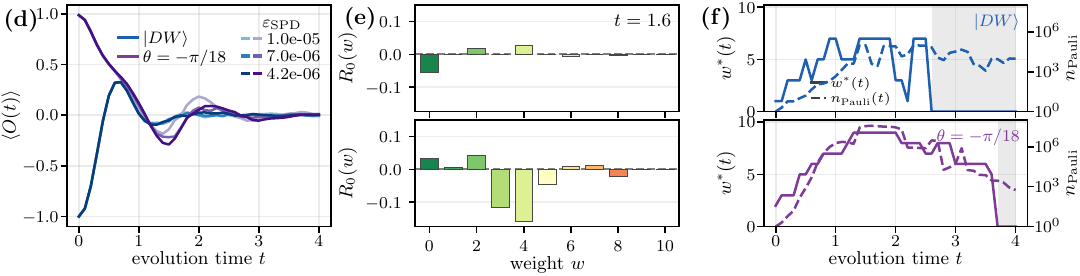}
 \caption{Numerical simulations of expectation values and reactivity functions using the SPD algorithm~\cite{begusic2025real}, for \textbf{(a-c)}~the observable $O=Z_{26}$ in the 1D MFIM on $N = 51$ qubits, and \textbf{(d-f)}~$O=Z_{3,3}$ in the 2D TFIM on a $5\times 5$ lattice. 
 \textbf{(a,d)}~Expectation value $\langle O(t) \rangle \equiv \bra{\psi} e^{iHt} O e^{-iHt} \ket{\psi}$ for two example initial product states $\ket{\psi}$ in each model (see SM \cite{supp} for definitions).
 Shading illustrates convergence with the SPD coefficient threshold \(\epsilon_{\text{SPD}}\); smaller thresholds correspond to more complex and more accurate simulations.
 \textbf{(b,e)}~The reactivity function \(R_{\tau = 0}(w)\) for each state, at a time chosen near the maximum memory cost of the SPD simulation. Green to red shading is an aid to the eye, and does not represent numerical data. 
 \textbf{(c,f)}~The tail weight \(w^{*}(t)\) (solid lines, left axis) and SPD memory cost \(n_{\rm Pauli}(t)\) (dashed lines, right axis, logarithmic scale) over time for each state. 
 The two quantities track one another remarkably closely across all times, states, and models, supporting the use of the reactivity as a proxy for the complexity of local information algorithms.
 Shaded regions mark the decayed regime,
$|\langle O(t)\rangle| \leq \epsilon_{\text{tail}}$, in which the
prediction of $\langle O(t)\rangle$ becomes trivial: $w^{*}(t)$ is no
longer informative and typically vanishes, while $n_{\rm Pauli}(t)$
continues to reflect operator structure that no longer contributes to
$\langle O(t)\rangle$~\cite{supp}.
 See main text for discussion, and the SM for complete definitions, robustness checks, and additional data~\cite{supp}.}
\label{fig:numerics_maintext}
\end{figure*}

\emph{Reactivity as a complexity metric for LI algorithms.}---Let us now
address the use of the reactivity as a practical complexity metric for LI
 algorithms.
We provide related rigorous guarantees in the SM~\cite{supp}, and in-depth studies of the physical behavior of the reactivity in a companion work~\cite{companion}.
We focus on the sparse Pauli dynamics (SPD) algorithm~\cite{begusic2025real} for concreteness, and expect qualitatively similar behavior for other LI algorithms.
The SPD algorithm
 time-evolves an observable, $O(t) = \sum_P c_P(t) P$, in the Pauli basis, and deletes Pauli strings whose coefficients are below a truncation threshold, $|c_P(t)| \leq \epsilon_{\text{SPD}}$, at each time step.
The classical cost of the algorithm is
summarized by the number of retained Pauli strings,
\(n_{\rm Pauli}(t)\), needed to achieve a desired target accuracy~\cite{supp}.

We consider two paradigmatic models
(Fig.~\ref{fig:numerics_maintext}): the 1D mixed-field
Ising model (MFIM), for which LI algorithms 
typically converge easily~\cite{begusic2025real,rakovszky2022dissipation}, and the 2D transverse-field Ising model
(TFIM), for which LI algorithms have failed in
experimentally relevant regimes~\cite{haghshenas2026digital}. 
In each model, we study the time-evolution of a local observable, $O = Z_{(N+1)/2}$, from two initial product states. 
We observe significant model and state-dependent variations in the convergence of the SPD algorithm.
For some ``easy'' models and states [MFIM, $\ket{i}^{\otimes N}$; Fig.~\ref{fig:numerics_maintext}(a)], SPD converges quickly at large thresholds, while for other ``hard'' instances [TFIM, $\ket{\theta = -\frac{\pi}{18}}$~\footnote{Here, we follow Ref.~\cite{haghshenas2026digital} and let $\ket{\theta} \equiv ( \cos(\theta/2) \ket{0} + \sin(\theta/2) \ket{1})^{\otimes N}$.}; Fig.~\ref{fig:numerics_maintext}(d)], SPD converges only at much smaller thresholds.

To explain these variations, we turn to the reactivity function.
For most models and states, we observe the highest reactivities at the initial time, $\tau = 0$~\cite{companion}.
The resulting functions, $R_0(w)$, are shown in Fig.~\ref{fig:numerics_maintext}(b,e), at times $t$ where the cost of SPD is near maximal.
We observe an array of interesting behaviors, some fluctuating quickly with $w$ and others  slowly, and, crucially, a strong qualitative correspondence with the SPD convergence.
In particular, the aforementioned ``easy'' example has a reactivity concentrated at weights $\lesssim 4$, while the ``hard'' example has support up to large weights $\approx 8$.
In general, we have not observed support at weights larger than $ \approx 10$ in our numerical simulations, due to the difficulty of converging the SPD algorithm for such systems.

To test this observed correspondence  systematically, we compare the memory cost $n_{\text{Pauli}}$ required by the SPD algorithm with the maximum weight $w_*$ at which the reactivity has support  [Fig.~\ref{fig:numerics_maintext}(c,f)].
The former is obtained by sweeping the SPD threshold, and choosing the largest value that achieves a fixed target accuracy~\cite{supp}.
The latter is obtained by choosing the smallest $w_*$ such that $| \sum_{w > w_*} R_0(w) | \leq \epsilon_{\text{tail}}$ for a small fixed $\epsilon_{\text{tail}}$~\cite{supp}.
We  average each quantity over small time windows, $\Delta t \approx 0.3$, to smooth over sharp features. 
We observe a surprisingly close  agreement between $w_*$ and the logarithm of  $n_{\text{Pauli}}$ across all models, states, and times considered.
The two quantities rise at similar rates, peak at similar times, and subsequently decay, in close correspondence with one another.
We note that this agreement is particularly remarkable as the SPD algorithm truncates operators based only on their coefficient magnitudes, and not weight.
Hence, our results provide  evidence that such algorithms can implicitly depend on the weight  even when not explicit in algorithm instructions. 
Our observation that the reactivity peaks and decays at late times also sharply contrasts with other complexity metrics. 

\emph{Learning and fast-forwarding.}---Beyond the classical simulation complexity, in the SM~\cite{supp},
we also provide rigorous evidence  that the reactivity controls the ease of \emph{learning} and \emph{fast-forwarding} quantum experiments.
For learning, we prove that, for any complete basis of states $\{\psi\}$,  $\tr(O\mathcal{C}(\psi))$ can be accurately predicted using local randomized measurement data if the mean-square reactivity is supported on small weights.
For fast-forwarding, we consider a circuit composed of two time-evolutions, $\mathcal{C} \equiv \mathcal{C}_2 \circ \mathcal{C}_1$, and prove that $\tr(O \mathcal{C}(\psi))$ can be accurately computed from randomized measurement data on $\mathcal{C}_1$ and $\mathcal{C}_2$ individually, under a similar assumption.
In both cases, the required number of randomized measurements scales exponentially in the maximum weight, $n^{\mathcal{O}(w_*)}$.

\begin{figure*}
  \centering
  \includegraphics[width=\textwidth]{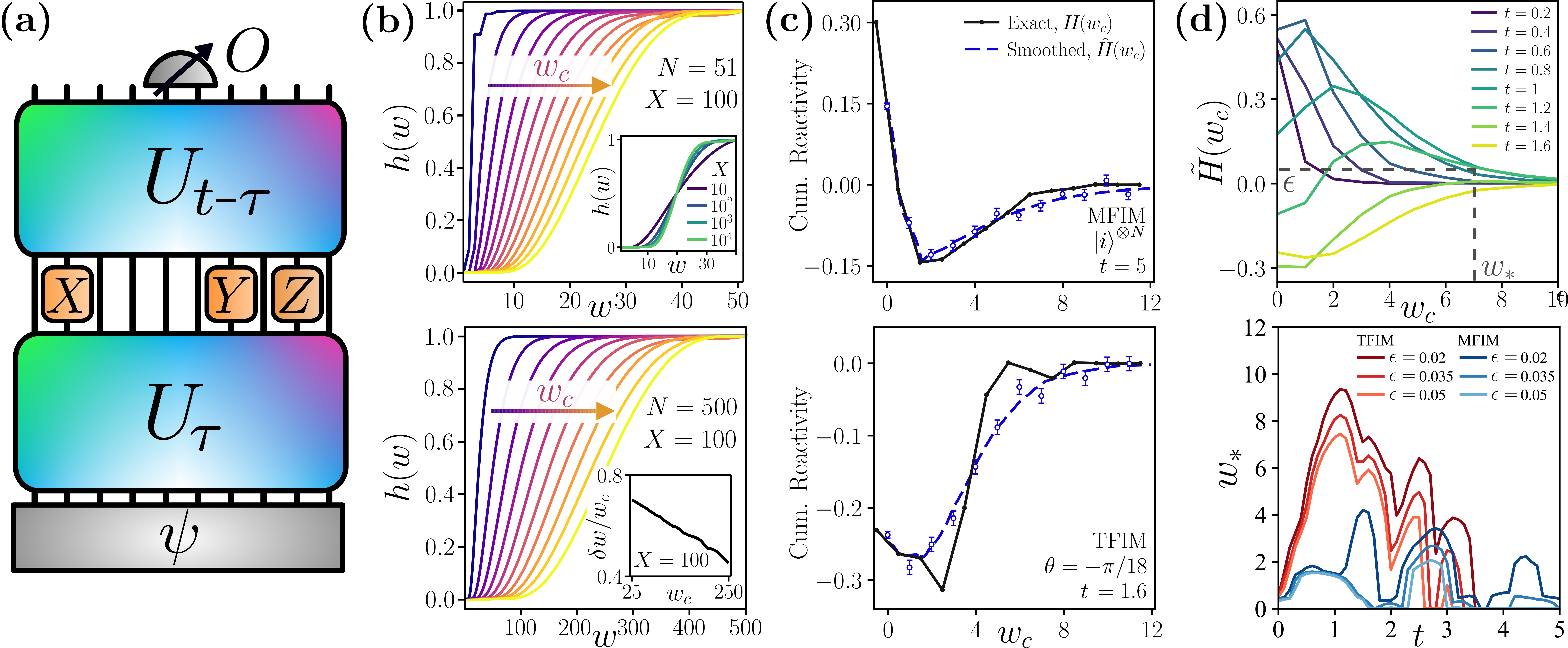}  
 \caption{\textbf{(a)} 
 The Pauli path spectroscopy protocol.
 To probe the reactivity function, $R_\tau(w)$, one inserts random Pauli operators of various weights $k$ at time $\tau$ of the experiment of interest.
 Taking a linear combination of the mean expectation value for each weight $k$ yields the inner product, $\sum_w h(w)R_\tau(w)$, of $R_\tau(w)$ with the filter function $h(w) = \sum_k h_k F(w,k;N)$ (see main text).
\textbf{(b)} To probe the total contribution of Pauli paths above a threshold weight $w_c$, we numerically optimize the filter functions to mimic a Heaviside function at weight $w_c$.
Top: Results at $N = 51$ qubits and $X=100$ sampling overhead; inset shows that larger sampling overheads enable sharper approximations. 
Bottom: Results at $N=500$ qubits and the same sampling overhead are similar and demonstrate scalability; inset shows that the function's relative width decreases mildly with $w_c$.
\textbf{(c)} Comparison of the exact cumulative reactivity function, $H(w_c) = \sum_{w > w_c} R(w)$ (black), and the ``smoothed'' cumulative reactivity function, $\tilde{H}(w_c) = \sum_w h(w;w_c) R(w)$ obtained from filter functions  with $X = 100$ sampling overhead (blue dashed lines denote exact $\tilde{H}(w_c)$; blue circles denote  estimated $\tilde{H}(w_c)$ assuming $10^6$ experimental shots per $w_c$). Models and initial states are described in Fig.~\ref{fig:numerics_maintext}.
\textbf{(d)} Examples of physical investigations enabled by our protocol. Top: The smoothed cumulative reactivity at various times in the TFIM, $\ket{\theta = -\pi/18}$.
Bottom: Time-evolution of the maximum weight $w_*$ at which $\tilde{H}(w_c)$ has support, for the MFIM, $\ket{i}^{\otimes N}$ (blue) and TFIM, $\ket{\theta=-\pi/18}$ (orange), for several precisions $|\tilde{H}(w_*)| \leq \epsilon$ (taking $10^7$ shots per $w_c$).}
\label{fig:3}
\end{figure*}

\emph{Experimental protocols.}---Having established the connection between the reactivity and the complexity of LI algorithms, let us now address its scalable experimental measurement~\cite{fn3}.
We remark that the existence of such a measurement protocol is highly non-trivial, as most complexity proxies require an exponential sampling overhead to observe~\cite{schuster2024random}.

Our protocol, \emph{Pauli path spectroscopy} [Fig.~\ref{fig:3}(a)], exploits the relation between the reactivity and noise.
In particular, our Eq.~(\ref{eq:reactivity-noise}) already implies that the noisy expectation value, $C(\gamma) \equiv \tr(O \mathcal{C}_\gamma(\psi))$, is the Laplace transform of $R(w)$.
This suggests that the reactivity can be recovered by an inverse Laplace transform, $R(w) = \sum_\gamma b_{w,\gamma} C(\gamma)$, for certain coefficients $b_{w,\gamma}$.
Unfortunately, the inverse Laplace transform is notoriously ill-conditioned~\cite{epstein2008bad}, and one finds that the coefficients $b_{w,\gamma}$ diverge exponentially in the weight $w$ that one is trying to recover.

To overcome this barrier, we ask a complementary question: What features of the reactivity function \emph{can} be efficiently recovered?
In particular, we observe that any linear combination of noisy expectation values,
\begin{equation} \label{eq:reactivity-measurement}
    \sum_\gamma h_\gamma \tr(O \, \mathcal{C}_\gamma(\psi)) = \sum_w h(w) R_\tau(w),
\end{equation}
can be efficiently estimated to precision $\varepsilon$ with a number of samples $(\sum_\gamma |h_\gamma|)^2/\varepsilon^2$~\cite{supp}.
To do so, one simply repeats the original experiment with local depolarizing noise of varying rates $\gamma$ inserted at time $\tau$, and takes a linear combination of their results.
Decomposing each noisy expectation value using Eq.~(\ref{eq:reactivity-noise}), one finds that this linear combination is equal to the overlap of the reactivity $R(w)$ with the \emph{filter function}, $h(w) \equiv \sum_\gamma h_\gamma e^{-\gamma w}$.
Hence, the overlap of $R(w)$ with any $h(w)$ with moderate-sized coefficients can be efficiently measured.

We find that these efficiently measurable $h(w)$ tend to be \emph{smooth}.
In particular, resolving a feature of $R(w)$ near weight $w$ with width $\delta w$ can generically be achieved with a sampling overhead $X \equiv (\sum_\gamma|h_\gamma|)^2 \sim e^{\mathcal{O}(w/\delta w)}$.
Crucially, as long as the relative width of the feature is order one, $\delta w/w = \mathcal{O}(1)$, the sampling overhead does not increase with $w$, and hence is scalable to large-weight, potentially beyond-classical, regimes.
We provide an example of such filter functions in Fig.~\ref{fig:3}(b), and systematic numerical and analytical studies in the End Matter.

In practice, we  find that it is more efficient to replace the inserted depolarizing channel with discrete single-qubit Pauli noise events.
In particular, if one repeats the experiment with random Pauli operators of weight $k$ inserted at time $\tau$, one can form the linear combination,
\begin{equation} \label{eq:reactivity-measurement 2}
    \sum_k h_k \tr(O \, \mathcal{C}_k(\psi)) = \sum_w h(w) R_\tau(w),
\end{equation}
with $h(w) \equiv \sum_k h_k F(w,k;N)$, where $\sum_k |h_k| \leq \sum_\gamma |h_\gamma|$, and $F(w,k;N) \approx (1-4w/3N)^k$~\cite{supp}.

Let us now illustrate how our protocol might be applied [Fig.~\ref{fig:3}(b-d)].
We begin with a primary motivation: to determine whether a given quantum experiment is dominated by local or non-local information.
To probe this, we design a family of filter functions, $h(w;w_c)$, that resemble smoothed Heaviside functions centered at $w = w_c$ [Fig.~\ref{fig:3}(b)].
The overlap of $R_\tau(w)$ with $h(w;w_c)$ thus probes the cumulative contribution to the experiment from weights $w \gtrsim w_c$. 
We numerically optimize the choice of $h_k$ to achieve the sharpest Heaviside approximation while maintaining a desired sampling overhead~\footnote{As shown in Fig.~\ref{fig:3}(b), we find that in practice, sampling overheads of $X \approx 10^2$ already yield reasonable smooth approximations to the Heaviside function. Increasing $X$ beyond this overhead yields diminishing returns.}.

We next proceed to the filter function's experimental measurement.
We numerically simulate our protocol using the SPD algorithm, with realistic shot noise added. 
The resulting ``smoothed'' cumulative reactivities, $\tilde{H}(w_c) \equiv \sum_w h(w;w_c) R(w)$, are shown in Fig.~\ref{fig:3}(c), alongside their exact counterparts, $H(w_c) \equiv \sum_{w > w_c} R(w)$.
We consider two representative models and states, identical to our earlier studies (Fig.~\ref{fig:numerics_maintext}): the 1D MFIM with  state $\ket{i}^{\otimes N}$, and the 2D TFIM with  state $\ket{\theta = -\pi/18}$.
For both examples, the smoothed reactivities display small shot-noise error, and resemble smooth interpolations of the exact reactivities, as anticipated.

Many further investigations are possible from the smoothed reactivity measurement data.
We display several examples in Fig.~\ref{fig:3}(c,d).
In Fig.~\ref{fig:3}(d), top, we show the smoothed cumulative reactivity of the TFIM for times $t$ increasing from zero until the breakdown of the SPD algorithm.
We observe a clear growth to larger-weight contributions leading up to this breakdown; one envisions that, in a quantum experiment, this growth could be characterized to much larger weights.
We also observe broad oscillatory features emerge in the smoothed reactivity, suggesting that its precise  form may be interesting to explore in larger-scale investigations.

To more compactly visualize the time-dependence of the smoothed reactivity, we can adopt our procedure from Fig.~\ref{fig:numerics_maintext}(c,f).
Namely, for each time $t$, we compute the weight $w^*$ above which $|\tilde{H}(w_c)| \leq \epsilon$ for a small $\epsilon$ [Fig.~\ref{fig:3}(d), bottom].
In keeping with our numerical observations, we observe that $w^*$ is nearly constant in time for the classically-easy MFIM, yet rapidly grows in time for the  TFIM.
In the latter, we also observe an intriguing dependence of the rate of this growth on the target error $\epsilon$; this connects to an important question in quantum simulation, of how quantum advantage scales with the demanded precision.

We conclude by emphasizing that our experimental approach, unlike our numerical simulations, is easily scalable to larger system sizes and weights $w_c$.
For our protocol with depolarizing noise, this is extremely simple: one re-scales the noise rates  $\gamma \rightarrow \gamma/x$, which replaces $w_c \rightarrow x w_c$ with no change in sampling overhead.
For our protocol with discrete noise events, we numerically find moderate efficiency \emph{gains} from working at larger $w_c$; examples are shown for $w_c$ up to $250$ in  Fig.~\ref{fig:3}(b).
This scalability suggests that the reactivity may be useful for benchmarking  complexity in quantum experiments that are beyond the reach of classical simulation:
Does an experiment reach ``only a little'' or ``very far'' beyond local-information classical algorithms?
And how does this answer scale with the evolution time, precision, system size, temperature and other characteristics of the system?

Finally,  we also introduce a second measurement protocol, which provides exact access to $R(w)$ for any infinite-temperature correlator (Fig.~\ref{fig:coherent})~\cite{supp}.
Our protocol borrows ideas from  information scrambling~\cite{schuster2022many} to replace the inserted noise with coherent rotations on a doubled system of EPR pairs with ancilla qubits.
This enables exact estimation of $R_\tau(w)$ from a Fourier transform.

\begin{figure}
\centering
\includegraphics[width=0.85\columnwidth]{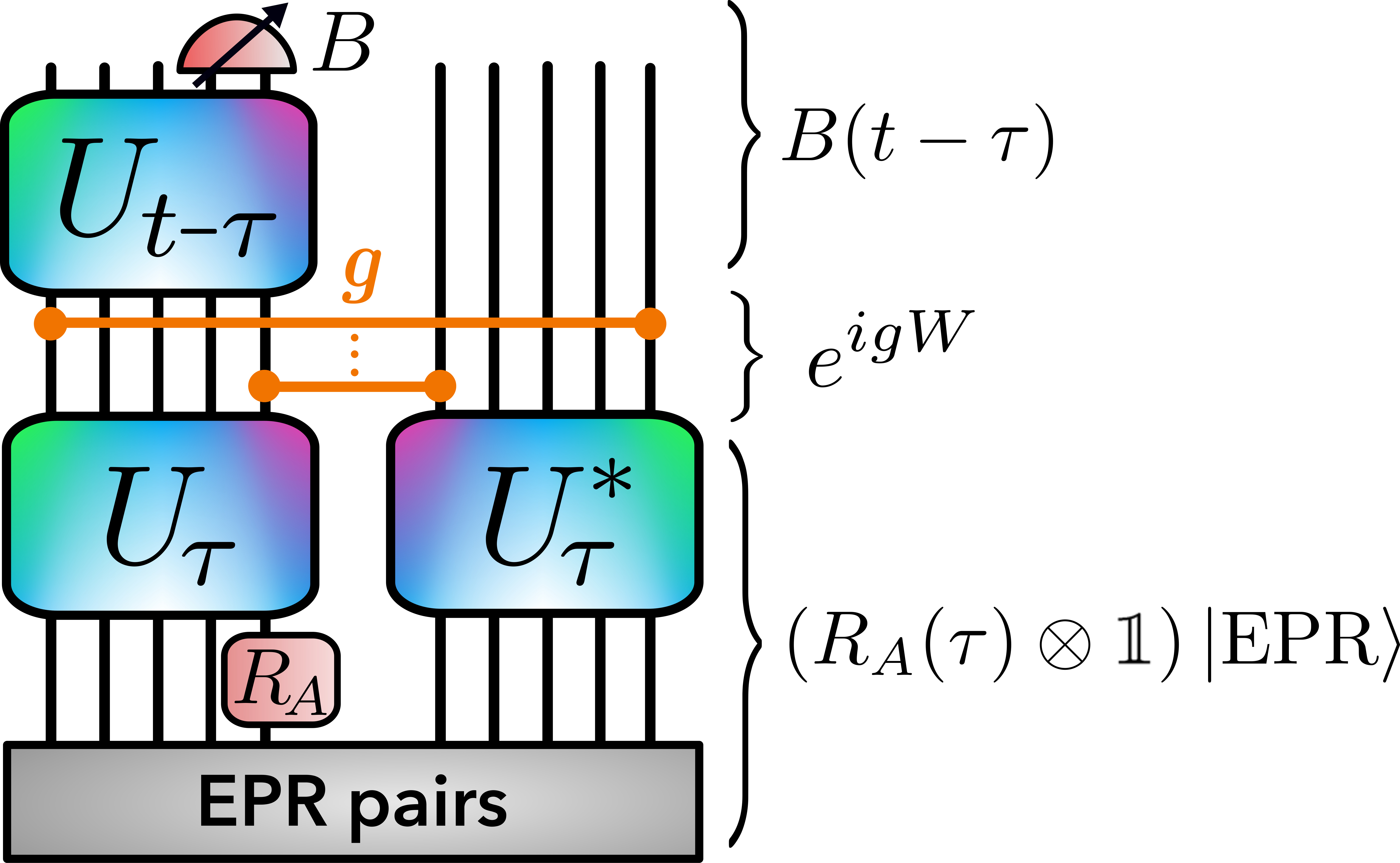}
\caption{Coherent protocol for Pauli path spectroscopy, which uses $N$ EPR pairs to measure the  Fourier transform of the reactivity, $\mathcal{R}_\tau(g) = \sum_w e^{igw} R_\tau(w)$, for any infinite-temperature correlator, $\overline{\tr}(B(t)A)$.
Here, $R_A \equiv (\identity+iA)/\sqrt{2}$ and $W \equiv \frac{3N}{4} -\frac{1}{4}\sum_{i=1}^N (X_{i,L} \otimes X_{i,R}-Y_{i,L} \otimes Y_{i,R}+Z_{i,L} \otimes Z_{i,R})$~\cite{supp}.}
\label{fig:coherent}
\end{figure}

\emph{Outlook.}---We have introduced a simple and measurable metric for the complexity of quantum dynamics experiments with respect to classical algorithms that track local information. 
As outlined  above, perhaps the most exciting direction opened by our work is the ability to pursue experimental investigations of quantum advantages at scales beyond the reach of classical algorithms.
There is little barrier to such study, as our protocol is easily implementable, and scalable to large sizes.

More broadly, our introduction of the reactivity itself raises several questions.
From a fundamental perspective, the reactivity represents a new probe of quantum information dynamics, with fundamentally different behaviors from conventional measures such as information scrambling.
Can one compute  reactivity function dynamics in representative models, or in systems with field-theoretic or hydrodynamic descriptions?
The  connection between the reactivity and local depolarizing noise may be  helpful in this regard.
Second, from a practical perspective, what applications might the reactivity have beyond benchmarking classical simulation complexity? 
As one example, a concurrent work building on our results uses the reactivity function to improve the understanding and design of quantum error mitigation methods~\cite{suchsland2026quantum}.
The physical link between the reactivity and noise suggests that related applications may continue to be found.

\textit{Acknowledgements}---We are grateful to Samantha Barron, Abhinav Kandala, Zoe Holmes, Hsin-Yuan Huang, Zlatko Minev, Thomas O'Brien, John Preskill, Tibor Rakovszky, Manuel Rudolph, Philippe Suchsland, Kristan Temme, Shreya Vardhan, Curt von Keyserlingk, Brayden Ware, and Norman Y. Yao for valuable discussions and insights.
We acknowledge the use of LLMs as an aid in numerical studies.
T.S. acknowledges
support from the Walter Burke Institute for Theoretical Physics at Caltech and the U.S. Department of Energy, Office of Science, National Quantum Information Science Research Centers, Quantum Systems Accelerator.
The Institute for Quantum Information and Matter is an NSF Physics Frontiers Center.

\let\oldaddcontentsline\addcontentsline
\renewcommand{\addcontentsline}[3]{}
\bibliography{refs}
\let\addcontentsline\oldaddcontentsline

\onecolumngrid
\newpage

\let\oldaddcontentsline\addcontentsline
\renewcommand{\addcontentsline}[3]{}
\section*{End Matter}
\let\addcontentsline\oldaddcontentsline

\onecolumngrid

\vspace{-2mm}
Here, we provide details on the construction of the smoothed filter functions described in the main text. We focus on filter functions that approximate either the Heaviside function or the delta function. We consider three approaches to constructing such functions: one based on the Chebyshev polynomials, one based on a particular simple analytic functional form, and one based on numerical optimization.

As discussed in the main text, we consider two forms of inserted noise: local depolarizing noise of strength $\gamma$ on each qubit, or application of random Pauli operators of weight $k$.
These yield the respective filter functions,
\begin{equation}
 h(w)=\sum_\gamma h_\gamma e^{-\gamma w},
 \qquad
 h(w)=\sum_k h_k \prescript{}{2}F_1^{}(-w,-k,-N;4/3),
 \vspace{-1mm}
 \label{eq:endmatter-responses}
\end{equation}
where $h_\gamma$ and $h_k$ are chosen coefficients, and $\prescript{}{2}F_1^{}(-w,-k,-N;4/3)$ is the mean damping of a weight-$w$ Pauli operator by a random weight-$k$ Pauli insertion (see SM~\cite{supp}).
The sampling overheads  are $X=(\sum_\gamma|h_\gamma|)^2$ and $X=(\sum_k|h_k|)^2$, respectively.
In general, larger and more signed coefficients enable more versatile filter functions, but also incur much larger sampling overheads.
We remark that any filter function using depolarizing noise can be implemented using noise insertions, with an equal or smaller sampling overhead, by setting $\tilde{h}_k \equiv \sum_\gamma p^\gamma_k h_\gamma$ where $p^\gamma_k = e^{-(3/4)\gamma N} ((3/4)\gamma N)^k/k!$ is the probability of $k$ noise events occurring in the depolarizing channel of strength $\gamma$ on $N$ qubits~\cite{suchsland2026quantum}.

\emph{Fourier-to-Chebyshev filters.}---Our first construction is based on a mapping to the Chebyshev polynomials.
This construction is the most convenient from an analytic perspective.
However, it tends to yield filter functions with undesirable oscillations in practice [Fig.~\ref{fig:filter-comparison}(a)], and so we do not use this method outside of our theoretical proofs.
We remark that our numerically optimized filter functions [Fig.~\ref{fig:filter-comparison}(c)] display a similar scaling of the sampling overhead with the inverse width as the Cheybyshev filter functions, but without these undesired oscillations.

Our method exploits a change of variables, from the weight $w$ to an angle $\theta$ defined by $\cos(\theta)=2e^{-\gamma_0 w}-1$, where $\gamma_0=\log(2)/w_*$ sets $\theta = 0$ at $w = w_*$ for convenience.
After this change of variables, we approximate the delta and Heaviside functions in $w$ (and hence, in $\theta$) by their truncated Fourier series, 
\vspace{-1mm}
\begin{align}
 h^{\mathrm{delta}}_n(\theta)
 & =
    \frac{1}{n}+\frac{2}{n}\sum_{m=1}^{(n-1)/2}(-1)^m\cos(2m\theta) =\frac{1}{n} \frac{\sin(n(\theta-\pi/2))}{\sin(\theta-\pi/2)}, \nonumber\\
 h^{\mathrm{Heaviside}}_n(\theta)
 &=\frac12-\frac{2}{\pi}\sum_{m=0}^{(n-1)/2}
   \frac{(-1)^m}{2m+1}\cos((2m+1)\theta).
\end{align}
respectively.
This can be converted to a power series in $z \equiv e^{-\gamma_0 w}$ using the Chebyshev polynomials, $T_m(2z-1)$, which obey $T_m(\cos\theta)=\cos(m\theta)$.
Explicitly, this yields
\begin{equation}
    h^{\text{delta}}_n(w;w_*) = 
    (-1)^{(n-1)/2} \frac{1}{n} \frac{T_n(2e^{-\gamma_0w}-1)}{2e^{-\gamma_0w}-1},
    \qquad
    h^{\text{Heaviside}}_n(w;w_*) = 
    \frac12-\frac{2}{\pi} \sum_{m=0}^{(n-1)/2}\frac{(-1)^m}{2m+1} T_{2m+1}(2e^{-\gamma_0w}-1).
\end{equation}
Each filter is a degree-$n$ polynomial in  $e^{-\gamma_0 w}$, and so can be measured with inserted noise rates $\gamma=0,\gamma_0,\ldots,n\gamma_0$. 

The Fourier-to-Chebyshev filter functions provide the cleanest analytic understanding of the dependence of the sampling overhead on the width $\delta w$.
The filter functions have width $\delta \theta \sim 1/n$, since they are obtained from a Fourier expansion up to frequency $n$.
This translates to $\delta w \sim \gamma_0^{-1}/n \sim w_*/n$.
Meanwhile, they have sampling overhead $X \sim e^{\mathcal{O}(n)}$ due to the exponential growth in the Chebyshev polynomial coefficients.
Hence, we have $X \sim e^{\mathcal{O}(w_*/\delta w)}$.

\emph{Monotonic analytic filters.}---Our second construction is an analytic construction that avoids the undesirable oscillations of the Chebyshev filter functions, at the cost of slightly worse scaling of the sampling overhead.
We introduce the following filter functions [Fig.~\ref{fig:filter-comparison}(b)], parameterized by the real number $r \in [0,\infty)$, 
\begin{equation}
 h^{\text{delta}}_r(w;w_*)=\frac{C_r}{w}
 \exp\!\left[-r\left(\frac{\kappa_rx}{w}
               +\frac{w}{\kappa_rx}\right)\right],
               \qquad
 h^{\text{Heaviside}}_r(x,w)= \int_{w_*}^\infty dw_*' h^{\text{delta}}_r(w;w_*').
 \label{eq:endmatter-monotonic}
\end{equation}
where $\kappa_r=\frac{K_2(2r)}{K_1(2r)}$, $C_r=\frac{\kappa_r}{2K_1(2r)}$, and $K_\nu(z)$ are the modified Bessel functions of the second kind.
For each $w$ and $r$, the first filter function is a normalized probability distribution over $w_*$ with mean $w$.
The distribution is tightly peaked about its mean,
with a width scaling as a small fraction, $\sim \! 1/2\sqrt{r}$, of the mean, $( \overline{w_*^2}-\overline{w_*}^2 )^{1/2} = \frac{w}{\sqrt{2r}} ( 1 + \mathcal{O}(1/r) )$.
The parameter $r$ controls the degree of smoothing under the filter function.
Larger $r$ (i.e.~less smoothing) will bring the filter function closer to the delta function, but will also incur a higher sampling overhead.

The precise functional form of $h^{\text{delta}}_r(w;w_*)$ was chosen with three ingredients in mind: 
(i) a width that scales as a tunable fraction of the mean, which we find is the natural scaling to keep the sampling overhead $X(w_*)$ non-increasing in $w_*$,
(ii) steep vanishing as $w,w_* \rightarrow 0$, to avoid unphysical support on negative weights, 
and (iii) a sufficiently simple expression to enable analytic computation of its inverse Laplace transform,
\begin{equation}
 h^{\text{delta}}_r(w,w_*)=\int_0^\infty d\gamma\,e^{-\gamma w}
 C_rJ_0\!\left(2\sqrt{r\kappa_r}\sqrt{w_*\gamma-r/\kappa_r}\right)
 \Theta(w_*\gamma-r/\kappa_r),
 \label{eq:endmatter-monotonic-inverse}
\end{equation}
where $J_\nu(z)$ is the Bessel function of the first kind.
This expression gives the filter function coefficients for inserted depolarizing noise.
Integrating the first filter function from $w_*$ to infinity gives our second filter function, which rises monotonically from zero to one and therefore approximates a Heaviside function.
The corresponding expression for its inverse Laplace transform involves $J_1$ in place of $J_0$.
The sampling overhead of each filter function grows as $e^{\mathcal{O}(r)}$, owing to the normalization coefficient $C_r$.
This is exponential in the \emph{square} of the inverse width, and hence sub-optimal compared to our Chebyshev and numerical filter functions at large $r$.
In practice, we truncate the integral over $\gamma$ at a large order one upper bound, chosen such that the truncation has a negligible effect at any $w \geq 1$.

\begin{figure*}[t]
  \centering
  \includegraphics[width=\textwidth]
    {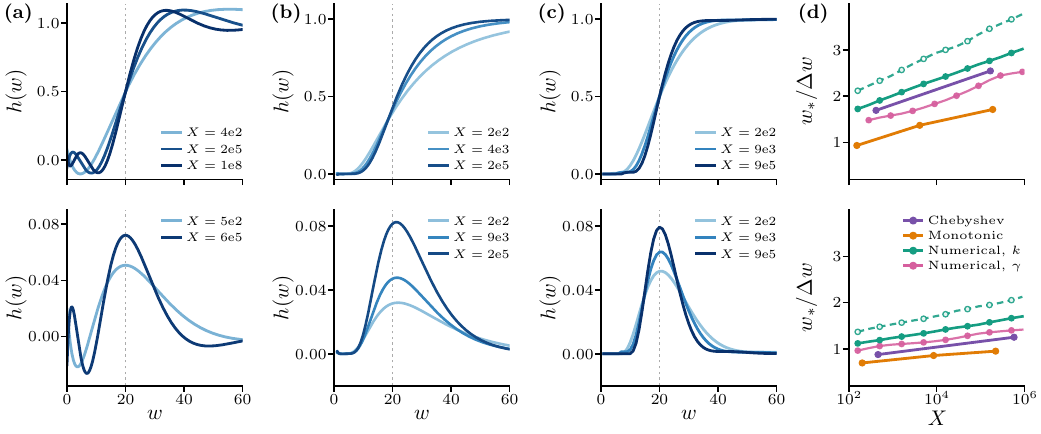}
  \caption{Smoothed filter functions that approximate the Heaviside function $\Theta(w-x)$ (top) and delta function $\delta(w-x)$ (bottom) for $x = 20$.
  \textbf{(a)} Fourier-to-Chebyshev filter functions, for degrees $n=3,5,7$
  (top) and $n=5,7$ (bottom).
  \textbf{(b)} Monotonic filter functions, for $r=1,2,3$ (top) and
  $r=1.75,2.5,3$ (bottom).
  \textbf{(c)} Numerically optimized filter functions, for $N=100$ using
  random Pauli insertions.
  In \textbf{(a-c)}, color denotes the sampling overhead
  $X=(\sum_r |h_r|)^2$, listed in each legend, using random Pauli insertions.
  \textbf{(d)} The inverse relative width $w_*/\Delta w$ versus $X$ for each filter function.
  Here, $\Delta w$ is equal to the FWHM for the delta filter functions, and analogous for the Heaviside filter functions.
  The numerically optimized filter functions use either random Pauli insertions ($k$) or
  noise rates ($\gamma$); open dashed and filled solid
  curves denote $N=50$ and $N=100$, respectively.}
  \label{fig:filter-comparison}
\end{figure*}

\emph{Numerically optimized filters.}---In practice, we find that numerical optimization of the filter functions can yield more desirable properties with smaller sampling overhead compared to our analytic constructions [Fig.~\ref{fig:filter-comparison}(c)].
To construct our numerical filter functions, we consider a fixed grid of noise rates $\gamma$ or Pauli insertion weights $k$, chosen sufficiently finely that the resulting filter function does not depend on the grid spacing.
For the first filter function, we then solve the convex optimization task of minimizing a weighted squared difference from the delta function, while requiring the filter function to remain positive, normalized, and rising monotonically toward a single peak at $w=w_*$ before falling monotonically from then on.
For the second filter function, we solve the convex optimization task of minimizing a weighted squared difference from the Heaviside function, while requiring the filter function to remain monotonically increasing between zero and one, and to equal $1/2$ at $w=w_*$. 
In both cases, the weights are chosen to penalize deviation farther from $w_*$ more heavily than those near $w_*$.
These optimization conditions can be solved efficiently and are observed to yield smooth filter functions without oscillating features.

\emph{Sampling overhead comparison.}---We compare the inverse width of each filter function versus its sampling overhead in Fig.~\ref{fig:filter-comparison}(d).
The numerically optimized filter functions perform best in all cases; for random Pauli insertions, the width also improves at smaller $N$.
For the Chebyshev and numerical filter functions, the observed data are consistent with an exponential scaling of the sampling overhead with the inverse relative width, $X \sim e^{\mathcal{O}(w_*/\delta w)}$.

\newpage
{\centering
\large\bfseries
Supplementary Material: Probing the classical complexity of quantum dynamics experiments
\par}

\tableofcontents

\section{The reactivity of quantum experiments} \label{sec: reactivity}

In this section, we provide a more detailed review of the Pauli path framework and our definition of the reactivity, as well as comparisons between the reactivity and other metrics for complexity in large quantum systems.

\subsection{Review of the Pauli path framework} \label{sec: framework}

We consider quantum experiments that measure an observable $O$ after performing time-evolution on an initial state $\psi$.
We will generally have in mind that the observable $O$ is local; however, this is not fundamental to the Pauli path framework nor the definition of the reactivity.
For simplicity, we focus on unitary time-evolution generated by a depth $T$ quantum circuit, $U = U_T U_{T-1} \cdots U_2 U_1$ and consider systems consisting of $N$ qubits with Hilbert space dimension $d=2^N$.
Our analysis readily generalizes to continuous-time Hamiltonian evolution, as well as evolution under noisy circuits or Lindbladian dynamics.

For any state $\psi$, circuit $U$, and observable $O$, the expectation value of the observable can be decomposed as a sum of so-called \emph{Pauli paths}~\cite{gao2018efficient,aharonov2023polynomial,schuster2025polynomial}, as
\begin{equation}
    \tr( O U \psi U^\dagger ) = \sum_{\vec{P}} A(\vec{P}; \psi, U, O).
\end{equation}
Here, each Pauli path, $\vec{P} = (P_0,\ldots,P_T)$, corresponds to a sequence of Pauli operators $P_t$ for each layer, $t = 0,\ldots, T$, of the circuit.
Each Pauli operator is a tensor product, $P = \bigotimes_{i=1}^n P_i$, of single-qubit Pauli operators, $P_i \in \{ \identity_i, X_i, Y_i, Z_i \}$, on each qubit of the system.
The contribution of a given Pauli path $\vec{P}$ is equal to the product,
\begin{equation}
    A(\vec{P}; \psi, U, O) = \tr( \psi P_0 ) \cdot \prod_{t=1}^T \overline{\tr}( P_t U_t P_{t-1} U_t^\dagger ) \cdot \overline{\tr}( P_T O ),
\end{equation}
of (i) the overlap of the first Pauli operator $P_0$ with the initial state $\psi$, (ii) the product of transition amplitudes from each instantaneous Pauli operator $P_{t-1}$ to its subsequent Pauli operator $P_t$, under the circuit layer $U_t$, and (iii) the overlap of the final Pauli operator $P_T$ with the observable $O$.
This expression is easily derived by inserting a resolution of the identity, $X = \sum_P \overline{\tr}(XP) P$ where $\overline{\tr}(\cdot) \equiv \frac{1}{2^n}\tr(\cdot)$, of the observable in the Pauli basis at every time step. 
This is analogous to the usual Feynman path integral, but over a complete basis of operators instead of states.

We remark that the contribution $A(\vec{P})$ of a Pauli path $\vec{P}$ can be large only if all three factors (i-iii) are large.
This sharply differs from traditional metrics of operator support, such as operator size distributions, out-of-time-order correlators, and operator entanglement entropies, which consider time-evolved observables independent of any initial state [i.e.~factors (ii-iii) but not (i)].
Similarly, it differs from traditional metrics of state support, such as quantum weight enumerators, magic, or entanglement entropies, which consider time-evolved states independent of any observable [i.e.~factors (i-ii) but not (iii)].
We provide a further discussion of these differences in later subsections.
The consideration of all three factors (i-iii) in combination is essential to capturing the behavior of quantum experiments, which necessarily involve both an initial state $\psi$ and a measured observable $O$.

\subsection{Review of connection to experimental noise} \label{sec: framework noise}

For completeness, let us also briefly recall the connection between the Pauli path framework and experimental noise.
This connection motivated the development of the Pauli path framework in previous works~\cite{gao2018efficient,aharonov2023polynomial,schuster2025polynomial}, and will be the inspiration for our measurement protocols for the reactivity introduced later on.

Consider the same circuit $U$ as in the previous subsection, but in which local depolarizing noise channels of strength $\gamma$ are applied to each qubit at each circuit layer.
For an individual qubit $i$, we write the action of a single local depolarizing channel on a qubit $i$ as $\mathcal{D}_{\gamma,i}(\rho) = e^{-\gamma} \rho + (1-e^{-\gamma}) \tr_i(\rho)$.
In the Pauli basis, this translates to $\mathcal{D}_{\gamma,i}(P) = e^{-\gamma} P$, for each  Pauli operator $P$ with non-identity support on qubit $i$, and $\mathcal{D}_{\gamma,i}(P) = P$, for each $P$ with identity support on qubit $i$.
Taking the  product over all qubits $i = 1,\ldots,n$ at a single circuit layer, $\mathcal{D}_\gamma = \prod_i \mathcal{D}_{\gamma,i}$, we find $\mathcal{D}_\gamma(P) = e^{-\gamma w[P]}$.
Here, as in the main text, the weight, $w[P] = \sum_{i=1}^n (1-\delta_{P_i, \identity})$, of the Pauli operator $P$ is defined as its total number of non-identity elements.
We are interested in the noisy circuit, $\mathcal{C}_\gamma = \mathcal{D}_\gamma \circ \mathcal{U}_T \circ \mathcal{D}_\gamma \ldots  \mathcal{D}_\gamma \circ \mathcal{U}_1 \circ \mathcal{D}_\gamma$, where $\mathcal{U}_t(\cdot) \equiv U_t (\cdot) U^\dagger_\tau$ denotes conjugation by $U_t$.

From the above, the expectation value in the noisy quantum circuit can immediately be written as~\cite{gao2018efficient,aharonov2023polynomial,schuster2025polynomial},
\begin{equation}
    \tr( O \mathcal{C}_\gamma( \psi ) ) = \sum_{\vec{P}} e^{-\gamma w[\vec{P}]} A(\vec{P}; \psi, U, O),
\end{equation}
where $w[\vec{P}] = \sum_{\tau=0}^t w[P_t]$ is the weight of the Pauli path $\vec{P}$, and $A(\vec{P}; \psi, U, O)$ are the Pauli path amplitudes of the ideal noiseless circuit.
Hence, the Pauli path amplitudes $A(\vec{P}; \psi, U, O)$ are intimately tied to the response of the quantum experiment to local depolarizing noise.
We remark that this connection can (trivially) be adapted to any unital single-qubit noise channels, by performing appropriate single-qubit rotations of the Pauli basis and minor modifications to the definition of the weight~\cite{schuster2025polynomial}.

\subsection{Definition of the reactivity} \label{sec: reactivity definition}

Let us now turn to our definition of the reactivity.
As discussed in the main text, our aim is to capture the extent to which a quantum experiment is influenced by local information versus non-local information.
This is naturally achieved within the Pauli path framework, as introduced in the previous subsections.

One can quantify the locality of a Pauli path in different ways, which lead to a different variants of the reactivity.
We focus on two variants in our work.
\begin{enumerate}
\item First, we quantify the locality of a Pauli path $\vec{P}$ via its \emph{weight} at a particular time $s$, 
\begin{equation}
    w[P_\tau] \equiv \sum_{i=1}^n (1-\delta_{P_{s,i},\identity}).
\end{equation}
From this, we define the \emph{reactivity function}, $R_\tau(w)$, of a quantum experiment as the total contributions from all Pauli paths with weight $w$ at time $s$,
\begin{equation}
    R_\tau(w; \psi, U, O) = \sum_{\vec{P}} \delta_{w[P_\tau], w} \cdot A(\vec{P}; \psi, U, O).
\end{equation}
We will often neglect the labels $\psi, U, O$ when they are clear from context.

\item Second, we quantify the locality of a Pauli path $\vec{P}$ via its summed weight over all times, 
\begin{equation}
    w[\vec{P}] = \sum_{\tau=0}^t w[P_\tau] = \sum_{\tau=0}^t \sum_{i=1}^n (1-\delta_{P_{s,i},\identity}).
\end{equation}
From this, we define the \emph{global reactivity function}, $R_g(w)$, as the total contribution from all Pauli paths with summed weight $w$,
\begin{equation}
    R_g(w; \psi, U, O) = \sum_{\vec{P}} \delta_{w[\vec{P}],w} \cdot A(\vec{P}; \psi, U, O).
\end{equation}
We will again neglect the labels $\psi, U, O$ when they are clear from context.
\end{enumerate}

Our two definitions of the reactivity closely resemble one another, and are likely governed by similar descriptions in physical settings.
Nonetheless, they are adapted to different motivations and experimental contexts.

Our first definition, $R_\tau(w)$, is designed to most closely capture the complexity of practical local information classical algorithms.
These algorithms typically store an approximation of a time-evolved state or operator at each fixed moment in time, and update this approximation time-step-by-time-step through a quantum circuit or dynamics.
In general, the complexity of such an algorithm would be determined by the limiting time step, i.e.~the step at which the Pauli paths yielding significant contributions are most non-local.
To identify this time, one could measure or compute the reactivity at a small handful of times $\tau$, and observe for which value of $\tau$ the reactivity is the largest.
In our numerical studies for this work, we find that the limiting time $\tau$ is usually the initial time, $\tau=0$, for the quantum many-body dynamics experiments we have considered; we study the slice dependence in more detail in \cite{companion}.
(For infinite-temperature correlation functions, the limiting time is instead typically $\tau = t/2$.)

Our second definition, $R_g(w)$, is best suited to capturing the impacts of physical noise in quantum experiments, which will in general accumulate over time.
In particular, if one performs a quantum circuit with local depolarizing noise of strength $\gamma$ at each qubit and circuit layer, one finds that the noisy expectation value is controlled by the global reactivity function, $\tr(O\mathcal{C}_\gamma(\psi)) = \sum_w e^{-\gamma w} R_g(w)$.
Our primary motivation for introducing $R_g(w)$ is that we believe it may be useful in future work, for example in quantum error mitigation or quantum simulation algorithms, in which understanding the response of an expectation value to experimental errors is invaluable to the interpretation and design of quantum experiments.
On a distinct note, we remark that $R_g(w)$ is also closely connected to existing rigorously-proven classical algorithms for simulating noisy quantum circuits~\cite{gao2018efficient,aharonov2023polynomial,schuster2025polynomial,gonzalez2025pauli}, which typically enumerate all low-weight Pauli paths in a circuit regardless of their configurations in space and time.
Unfortunately, such algorithms are often less useful from a numerical standpoint, owing to a large practical number of terms in this enumeration.

Finally, while we have focused on these two definitions of the reactivity in our work, further adaptations are easily possible.
As one example, one could define a \emph{temporal reactivity function}, $R_x(w) = \sum_{\vec{P}} \delta_{w_x[\vec{P}],w} \cdot A(\vec{P})$, which captures the total contribution from all Pauli paths with weight $w$ on a fixed \emph{spatial} slice $x$, $w_x[\vec{P}] = \sum_{\tau=0}^t (1-\delta_{P_{t,x}})$.
We speculate that this might be useful for describing the complexity of transverse tensor network contraction methods~\cite{berezutskii2025tensor}, which simulate a quantum circuit qubit-by-qubit rather than layer-by-layer.
As a second example, one could define a \emph{local reactivity parameter}, $R_{x,s} = \sum_{\vec{P}} (1-\delta_{P_{s,x},\identity})$, which captures the total contribution from all Pauli paths with non-identity support on a given qubit $x$ at a given time $s$.
We speculate that this could serve as a useful one-number metric for the importance of a particular circuit gate (i.e.~at location $x$ and time $s$) within a given experiment.

\subsection{Comparison to operator size distributions}

In this section, we review a second, existing, measure of operator complexity, the \emph{operator size distribution}~\cite{roberts2018operator,schuster2022operator}.
Our purpose in reviewing this is (i) to emphasize the fundamental difference between operator size distributions, which also sort quantum information by the Pauli weight, and the reactivity function, and (ii) because our experimental protocols for measuring the reactivity immediately extend to measuring the operator size distribution as well (see later sections in the Supplementary Material).

In physics literature, the weight of a Pauli operator is often referred to as its \emph{size}.
To define the size distribution of a time-evolved operator, let us first decompose the operator as a sum of Pauli operators, 
\begin{equation}
    U^\dagger O U = \sum_{P} a(P; U, O) \, P.
\end{equation}
The amplitudes, $a(P; U, O)$, are normalized, in the sense that $\sum_P | a(P; U, O) |^2 = \overline{\tr}( O^\dagger O )$, where $\overline{\tr}(\cdot) \equiv \tr(\cdot) / 2^n$ is the normalized trace. 
We assume without loss of generality that $O$ is normalized, so that $\sum_P | a(P; U, O) |^2 = \overline{\tr}( O^\dagger O ) = 1$.
The size distribution, $P(w)$, of the operator is given by,
\begin{equation}
    P(w; U, O) = \sum_{P \, : \, w[P] = w} | a(P; U, O) |^2 = \overline{\tr}\big( U^\dagger O U \cdot \mathcal{P}_w [ U^\dagger O U ] \big).
\end{equation}
As for the reactivity, we will often suppress the labels $U,O$ when clear from context.
The size distribution is a normalized probability distribution, obeying $P(w) \geq 0$ and $\sum_w P(w) = 1$.
This contrasts with the reactivity function, which is neither positive nor (in general) normalized to one.

Let us now address purpose (i).
Unlike the reactivity, the operator size distribution is defined independently of any initial state $\psi$.
In practice, this means that the operator size distribution provides a reasonable proxy for the complexity of \emph{exact} simulation of a time-evolved operator $U^\dagger O U$, in the absence of any initial state.
This is in contrast to the reactivity function, which provides a proxy for the complexity of \emph{approximate} simulations, once the operator's expectation value is taken in an initial state $\psi$.

Let us now address purpose (ii).
The application of our experimental protocols to measuring size distributions follows because the size distribution can be viewed as the reactivity function of the Loschmidt echo circuit, or ``mirror circuit'', at its middle layer.
To see this, let $(V_{2t},\ldots,V_{t+1},V_t,\ldots,V_0) = (U^\dagger_1,\ldots,U^\dagger_{t}, U_t,\ldots,U_1)$ denote the mirror circuit and consider the initial state, $\psi_O = (\identity + O)/2^n$ (we assume that $O$ is a Pauli operator for simplicity).
We then have $P(w; U, O) = R_\tau(w; \psi_O, V, O)$.
This explains the size distribution's normalization to one, i.e.~the value of the Loschmidt echo for  unitary dynamics.
It also explains its positivity: in the Loschmidt echo, the ``initial state'' can be viewed as being proportional to the operator $U^\dagger O U$ itself, thus the contribution of each Pauli operator is squared and hence positive.

\subsection{Comparison to other complexity metrics}
\label{sec:appendix-othermetrics}

In the previous subsection, we saw that the operator size distribution could be interpreted as a proxy for the complexity of performing exact operator time-evolution using LI algorithms, while the reactivity function could be interpreted as a proxy for the complexity of performing  time-evolution that only seeks to the recover the expectation value outcome of a given quantum experiment.
In this subsection, we briefly remark that a similar contrast holds with many other complexity measures as well.
For example, the entanglement entropy, magic, and quantum weight enumerators~\cite{rains1998quantum,miller2026experimental} all seek in different ways to characterize the complexity of a given quantum state, in the absence of any observable.
Meanwhile, the operator entanglement entropy, operator magic, and operator size distribution all seek to characterize the complexity of a given quantum operator, in the absence of any initial state.
Other complexity metrics, such as the temporal entanglement entropy~\cite{hastings2015connecting,foligno2023temporal,vilkoviskiy2025temporal}, cannot cleanly be sorted as characterizing either an individual operator nor state, yet seem to display  qualitative features in common with other measures above, such as extensive values in random quantum circuits in the absence of noise~\cite{foligno2023temporal,vilkoviskiy2025temporal}.

In contrast to those listed above, several other more recent complexity measures \emph{do} bear a close connection and similar physics to the reactivity.
The most notable examples are the operator backflow~\cite{rakovszky2022dissipation,von2022operator} and effective quantum volume~\cite{kechedzhi2023effective}, whose connections to the reactivity are discussed in the main text . 
Other more specialized approximate measures of complexity have been introduced in specific contexts.
For example, out-of-time-order correlator experiments explicitly foil local information classical simulation algorithms, by refocusing large Pauli strings using time-reversal prior to measurement~\cite{xu2022scrambling,mi2021information,abanin2025constructive}.
This yields a very high reactivity by design.
Nevertheless, one finds that such experiments can sometimes be susceptible to more advanced classical simulation algorithms, for example, ones that approximate Pauli path dynamics using a classical probabilistic process with local quantum corrections added in~\cite{abanin2025constructive}.
The complexity of such classical algorithms is related to the typical locality of the relevant quantum corrections~\cite{abanin2025constructive}, in a manner broadly reminiscent to the reactivity function.
For these reasons, we hope that the central ideas of our work, i.e.~the formal definition of the reactivity function, and our experimental protocols to measure it, might also bear fruit in such other contexts as well.
Indeed, the connection between locality and complexity is found in many contexts, as is also the connection between locality and response to perturbations.

\section{Numerical results: further details}
\label{sec:appendix-numerics}

In this section, we present further details and additional results from our numerical studies of the reactivity function and the complexity of the sparse Pauli dynamics (SPD) algorithm.

\subsection{Models and settings}
\label{sec:appendix-models-settings}

We study two paradigmatic spin models. The one-dimensional mixed-field
Ising model (MFIM) on a chain of $N=51$ sites,
\begin{equation} \label{eq:H-mfim}
  H_{\rm MFIM} = \sum_{i} Z_i Z_{i+1}
  + g_x \sum_i X_i + g_z \sum_i Z_i,
\end{equation}
with $g_x = 1.4$ and $g_z = 0.9045$, a standard nonintegrable parameter
choice; and the two-dimensional transverse-field Ising model (TFIM) on
a $5\times5$ square lattice with open boundaries,
\begin{equation} \label{eq:H-tfim}
  H_{\rm TFIM} = J \sum_{\langle ij\rangle} Z_i Z_j
  + g_x \sum_i X_i,
\end{equation}
with $J=-1$ and $g_x = 2$~\cite{haghshenas2026digital}.

All simulations evolve the local observable $O = Z_i$ on the center
site in the Heisenberg picture ($Z_{25}$ and $Z_{3,3}$, respectively), using sparse Pauli dynamics (SPD) \cite{begusic2025real} with
second-order Trotterization at step $\Delta t = 0.05$ (MFIM) and
$\Delta t = 0.01$ (TFIM), coefficient truncation at threshold
$\epsilon_{\text{SPD}}$, and readout interval $0.1$ over the time windows
$t \in [0,10]$ (MFIM) and $t \in [0,4]$ (TFIM). For each model the
evolution is repeated over a sweep of truncation thresholds
$\epsilon_{\text{SPD}}^{(0)} \le \epsilon_{\text{SPD}}^{(1)}  \le \cdots$, with $\epsilon_{\text{SPD}}^{(0)} $ the
smallest, ranging from $4.21 \cdot 10^{-6}$ to $ 0.01$ for TFIM and from  $4 \cdot 10^{-6}$ to $0.01 $ for MFIM, respectively.

The initial states are product states. For the MFIM we use six
states: the four translation-invariant states
$|+\rangle^{\otimes N}$, $|{+}i\rangle^{\otimes N}$,
$|0\rangle^{\otimes N}$, and $|1\rangle^{\otimes N}$, and the two
domain-wall states
$|{+}\rangle^{\otimes 25}|{-}\rangle^{\otimes 25}|{+}\rangle$ and
$|0\rangle^{\otimes 25}|1\rangle^{\otimes 25}|0\rangle$, in which a
domain of $25$ flipped sites is embedded in the $N = 51$ site chain,
with the measured site at its first site. 
For the TFIM we use eight
states: the seven tilted states
$\ket{\theta_k} \equiv \bigotimes_i \big(\cos\tfrac{\theta_k}{2}\ket{0}_i
+ \sin\tfrac{\theta_k}{2}\ket{1}_i\big)$~\cite{haghshenas2026digital},
with $\theta_k = (k-3)\,\pi/18$ for $k = 0,\dots,6$, where the
extremal angle $\theta_0 = \arcsin(g_x J/4) = -\pi/6$ coincides with
the mean-field ground-state tilt of the infinite lattice; and the
$Z$-basis domain-wall state
$|0\rangle^{\otimes 12}|1\rangle^{\otimes 12}|0\rangle$, in which a
domain of $12$ flipped sites is embedded in the $5\times 5$ lattice
with respect to its row-major ordering. The analogous $X$-basis
domain-wall state shows no dynamics---it is an eigenstate of the
global spin flip $\prod_i X_i$ conserved by $H_{\rm TFIM}$, so that
$\langle O(t)\rangle$ vanishes identically---and we do not consider
it further.

The main text shows two states per model: $|+\rangle^{\otimes N}$
and $|{+}i\rangle^{\otimes N}$ for the MFIM, and the tilted state
$\theta = -\pi/18$ and the $Z$-basis domain-wall state for the TFIM
[Fig.~\ref{fig:numerics_maintext}].
Sec.~\ref{sec:appendix-states} extends the diagnostics to the full
state sets.

\subsection{Definitions of the diagnostics}
\label{sec:appendix-definitions}

\emph{Tail edge $w^{*}(t)$.}
For a fixed SPD threshold $\epsilon_{\text{SPD}}=\epsilon_{\text{SPD}}^{(0)}$, we define the net signed tail above weight $w$,
\begin{equation} \label{eq:S}
  H(t,w) \equiv \Big|\sum_{w'\ge w} R_0(w',t)\Big|,
\end{equation}
so that $H(t,0)=|\langle O(t)\rangle|$. Because contributions of
opposite sign cancel, $H(t,w)$ is not monotone in $w$. For a tolerance
$\epsilon_{\rm tail}>0$ the tail edge is the smallest weight beyond
which the net tail remains below tolerance at every higher weight,
\begin{equation} \label{eq:wstar}
  w^{*}_{\epsilon_{\rm tail}}(t) \equiv
  \min\{\, w \,:\, H(t,w') \le \epsilon_{\rm tail}
  \;\;\forall\, w' \ge w \,\}.
\end{equation}
The contiguous-from-the-top requirement selects the weight at which
the net tail becomes, and stays, negligible, rather than an isolated
weight at which an accidental cancellation dips below tolerance. We
use $\epsilon_{\rm tail}=10^{-2}$; if no such $w$ exists, $w^{*}$
is set to the largest weight resolved. To desensitize the result to
the exact tolerance we report the median over a narrow band,
\begin{equation} \label{eq:wstar_jitter}
  w^{*}(t) = \mathrm{med}\big\{\,
  w^{*}_{c\,\epsilon_{\rm tail}}(t) : c \in \{0.9,1.0,1.1\}\,\big\},
\end{equation}
in which we multiply the tail tolerance by $c = 0.9,1.0,1.1$, respectively.

\emph{Memory cost $n_{\rm Pauli}(t)$.}
SPD propagates the Heisenberg operator independently of the initial
state, storing $N_{\rm Pauli}(\epsilon_{\rm SPD},t)$ strings; the state enters
through the accuracy at which a threshold reproduces
$\langle O(t)\rangle$. Let
$f(t,\epsilon_{\rm SPD})=\sum_w R_0(w,t;\epsilon_{\rm SPD})=\langle O(t;\epsilon_{\rm SPD})$ denote the truncated
expectation value and $\bar f(t)$ the median of $f$ over the
$m_{\rm ref}=3$ smallest thresholds, forming our ground truth. With
$D(t,\epsilon)=|f(t,\epsilon_{\rm SPD})-\bar f(t)|$, we select the loosest
threshold connected to the tightest by an unbroken run of agreement
within a band $\delta$,
\begin{equation} \label{eq:epsstar}
  \epsilon^{*}_{{\rm SPD},\delta}(t) \equiv
  \max\{\, \epsilon_{\rm SPD} \,:\, D(t,\epsilon_{\rm SPD}') \le \delta
  \;\;\forall\, \epsilon_{\rm SPD}' \le \epsilon_{\rm SPD} \,\},
\end{equation}
which guards against a coarse threshold that agrees with the reference
by accident (e.g.\ at zero-crossings of $\langle O(t)\rangle$, or at
late times where a maximally truncated operator predicts zero). The
reported cost is
$n_{\rm Pauli}(t) = N_{\rm Pauli}(\epsilon^{*}_{{\rm SPD},\delta}(t),t)$ with
$\delta = 0.1$, again taken as the median over a narrow band of $\delta$,
\begin{equation} \label{eq:npauli_jitter}
  n_{\rm Pauli}(t) = \mathrm{med}\big\{\,
  N_{\rm Pauli}(\epsilon^{*}_{{\rm SPD},c\delta}(t),t) :
  c \in \{0.9,1.0,1.1\}\,\big\}.
\end{equation}

\emph{Time window.}
To suppress observed sharp oscilations as a function of time, $H(t,w)$ and
$D(t,\epsilon_{\rm SPD})$ are replaced by their root-mean-square values over a
sliding window of width $\Delta t_w = 0.3$ (three readout points)
before Eqs.~\eqref{eq:wstar} and~\eqref{eq:epsstar} are applied;
$\Delta t_w = 0$ recovers the per-time definitions.

\emph{Decayed regime.}
The tail edge is informative only while the observable is resolvable
above the tail tolerance $\epsilon_{\text{tail}}$. In all time-resolved figures we shade the
late-time region in which the windowed $|\langle O(t)\rangle|$ remains
below $\epsilon_{\rm tail}$, determined with the same
contiguous-from-the-top rule as Eq.~\eqref{eq:wstar}: shading begins at
the first time, after the maximum of $|\langle O(t)\rangle|$, from
which the windowed observable stays below tolerance. There the net
tail falls below tolerance at every weight, prediction of
$\langle O(t)\rangle$ becomes trivial, and $w^{*}(t)$ typically vanishes, while
$n_{\rm Pauli}(t)$ can continue to reflect operator structure of $O(t)$ that does not contribute to the expectation value $\langle O(t)\rangle$. Where $w^{*}$ remains transiently
finite inside the shaded region, the small expectation value arises
from cancellation between low- and high-weight contributions rather
than from an absent tail.

\subsection{Robustness of the diagnostics}
\label{sec:appendix-robustness}

The diagnostics of Sec.~\ref{sec:appendix-definitions} carry four
numerical parameters: the SPD truncation threshold $\epsilon_{\text{SPD}}$ at
which the reactivity is computed, the tail tolerance
$\epsilon_{\rm tail}$ defining $w^{*}(t)$, the accuracy band
$\delta$ selecting $n_{\rm Pauli}(t)$, and the Trotter step
$\Delta t$ of the underlying evolution. Here we verify that the
correspondence between $w^{*}(t)$ and $n_{\rm Pauli}(t)$ reported in
the main text does not hinge on any of these choices.
Figure~\ref{fig:appendix-robustness} examines each parameter in turn
for a representative hard state of the TFIM, the tilted state
$\theta = -\pi/18$ of Fig.~\ref{fig:numerics_maintext}; the remaining
main-text states behave analogously (not shown).
We emphasize the distinction between the narrow $\pm10\%$ jitter
medians of Eqs.~\eqref{eq:wstar_jitter}
and~\eqref{eq:npauli_jitter}, which are part of the definitions and
regularize discretization artifacts, and the wide
$\times 0.5$--$\times 2$ sweeps displayed here, which probe the
sensitivity of the diagnostics to the overall tolerance scale.
Accordingly, panels~(b,c) are computed without the jitter medians,
so that the bare sensitivity is visible; at unit scale factor they
reproduce the baseline values $\epsilon_{\rm tail} = 10^{-2}$ and
$\delta = 0.1$, differing from the main-text curves only by the
omitted median. The sweeps in panels~(a--c) are pure re-analyses of
fixed SPD data, so no run-to-run variation enters; only panel~(d)
involves separate simulations. In panels~(b--d), shading marks the
decayed regime of Sec.~\ref{sec:appendix-definitions}, determined at
the baseline tolerance.

\paragraph{SPD truncation threshold.}
The reactivity $R_{0}(w;\epsilon)$ is itself computed by SPD, so
Pauli paths with coefficients below $\epsilon$ are discarded
before the tail is formed; the tail edge is then evaluated at the
reference threshold $\epsilon_{\rm ref} = \epsilon_{\rm SPD}^{(0)}$
(Sec.~\ref{sec:appendix-definitions}). Panel~(a) shows $|R_{0}(w)|$
at fixed $t = 2.0$ for five thresholds spaced by factors of
$\approx 2$. The finer thresholds coincide up to intermediate
weights and separate only in the deep tail, which extends to larger
weights as $\epsilon_{\rm SPD}$ decreases; the coarsest threshold shows
visible deviations already at intermediate weights. This is
the implicit weight dependence of LI algorithms discussed in the
main text, visible directly in the data: although SPD truncates on
coefficient magnitude alone, raising the threshold depletes the
reactivity from the high-weight end, so that coefficient truncation
acts, in effect, as a truncation on large weights. At this time the
tail edge [$w^{*}(t=2.0) = 9$, cf.~panel~(b)] lies within the region
over which the finer thresholds agree.
Where the edge instead approaches the largest weight retained, the
true tail extends beyond the available SPD resolution and
$w^{*}(t)$ should be read as a lower bound; the high-$w^{*}$
plateaus reported here are correspondingly conservative.

\paragraph{Tail tolerance $\epsilon_{\rm tail}$.}
Panel~(b) recomputes $w^{*}(t)$ with the tolerance scaled by factors
$\{0.5,\,2^{-1/2},\,1,\,2^{1/2},\,2\}$. The curves are nested by
construction, since a looser tolerance can only lower the contiguous
edge of Eq.~\eqref{eq:wstar}. Through the rise and plateau they
remain within two to three weight units of one another, and the
rise--plateau--decay structure is preserved. Approaching the decayed
regime, the curves collapse to zero at slightly different,
tolerance-ordered times, while showing qualitatively the same behavior. 

\paragraph{Accuracy band $\delta$.}
Panel~(c) recomputes $n_{\rm Pauli}(t)$ with the band $\delta$
scaled by the same factors. Enlarging $\delta$ can only raise the
selected threshold $\epsilon^{*}_{\rm SPD, \delta}(t)$ of
Eq.~\eqref{eq:epsstar}, so looser bands $\delta$ certify coarser runs with typically
smaller counts. The spread is strongly time-dependent. Through the
rise and at the peak, where the cost is large, the curves agree to
within about half a decade; in the late-time decay the spread grows
to two to three decades. This is a late-time effect: as
$\langle O(t)\rangle$ decays, predicting it to within $\delta$
approaches triviality, and progressively coarser thresholds satisfy
Eq.~\eqref{eq:epsstar} depending on $\delta$. The accuracy-matched
cost is thus intrinsically less sharply defined approaching the
decayed regime, and the comparison of $w^{*}(t)$ and
$\log n_{\rm Pauli}(t)$ there should be read with this in mind. The
rise--plateau--decay profile is unchanged at every $\delta$.

\paragraph{Trotter step $\Delta t$.}
Panel~(d) compares both diagnostics across
$\Delta t \in \{0.01,\,0.023,\,0.037,\,0.05\}$; the evolved
expectation value itself agrees across steps (not shown), so any
$\Delta t$-dependence of the diagnostics reflects the SPD
representation rather than the simulated physics. Both diagnostics
prove remarkably stable. The tail edge is essentially independent of
$\Delta t$ at all times: it is a structural property of the
Heisenberg-evolved operator, set by the underlying continuous-time
dynamics rather than by its discretization. The memory cost is
likewise nearly $\Delta t$-independent where it is large: through
the rise and at the peak, the accuracy-matched counts nearly
coincide, even though the number of strings generated at a fixed
truncation threshold differs strongly between step sizes. The
accuracy matching of Eq.~\eqref{eq:epsstar} thus largely removes the
dependence on the discretization, returning a cost characteristic of
the simulated dynamics rather than of its representation. Only in
the late-time decay do the counts separate, with finer steps
retaining more strings---the same window in which the
accuracy-matched cost is intrinsically less sharply defined, as
discussed above.

\paragraph{Implications.}
Across all four sweeps, the diagnostics are robust wherever the
observable is well resolvable and the cost is large: the tail edge, the
rise--plateau--decay profile of the memory cost, the co-movement of
the two diagnostics, and the peak cost itself are stable under
$\mathcal{O}(1)$ rescalings of the tolerances and a factor-of-five
change in the Trotter step. The residual sensitivity of both
diagnostics is confined to the late-time decay, where the prediction
target approaches triviality and the accuracy-matched cost is
intrinsically less sharply defined. Absolute counts in that window
remain both tolerance- and step-dependent and should not be
over-interpreted.

\begin{figure*}[t]
  \centering
  \includegraphics[width=\textwidth]{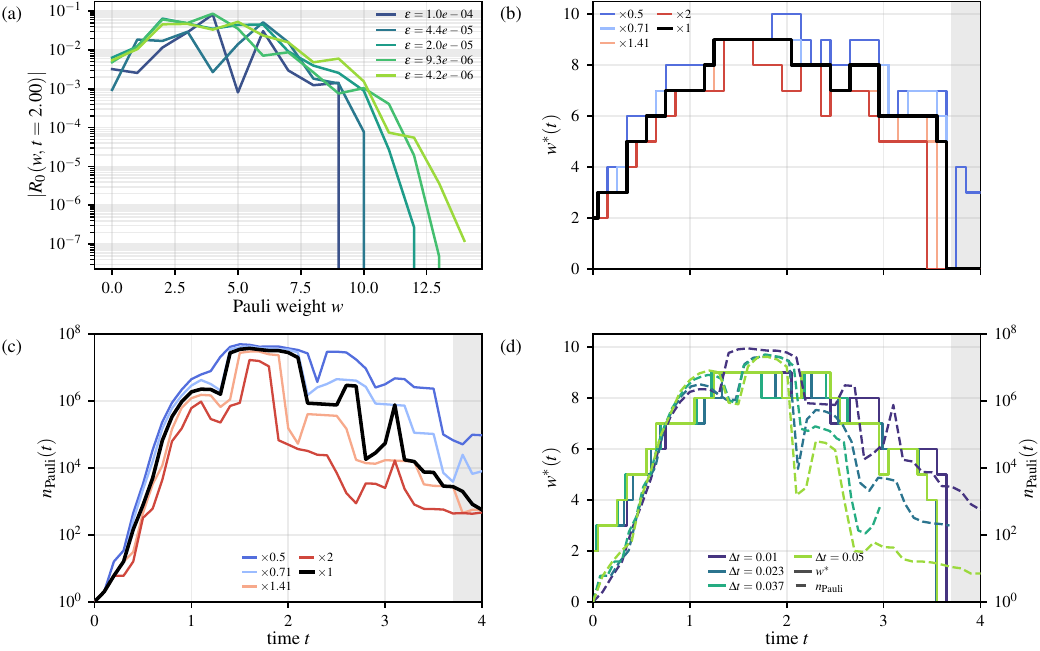}
  \caption{\textbf{Robustness of the diagnostics} (TFIM, tilted state
    $\theta = -\pi/18$, $\Delta t = 0.01$ unless stated).
    \textbf{(a)}~Reactivity $|R_{0}(w)|$ at $t={2.0}$ for five SPD
    truncation thresholds between $\epsilon_{\rm SPD} = 4.2\times10^{-6}$
    and $1.0\times10^{-4}$.
    \textbf{(b)}~Tail edge $w^{*}(t)$ at
    $\epsilon_{\rm tail}\times\{0.5, 0.71, 1, 1.41, 2\}$.
    \textbf{(c)}~Memory cost $n_{\rm Pauli}(t)$ at
    $\delta\times\{0.5, 0.71, 1, 1.41, 2\}$.
    \textbf{(d)}~$w^{*}(t)$ (solid, left axis) and $n_{\rm Pauli}(t)$
    (dashed, right axis, log scale) across Trotter steps
    $\Delta t \in \{0.01, 0.023, 0.037, 0.05\}$.
    Panels (b,c) are computed without the jitter medians of
    Eqs.~\eqref{eq:wstar_jitter} and~\eqref{eq:npauli_jitter}, so that
    the bare sensitivity is visible; in (b--d), shading marks the
    decayed regime of Sec.~\ref{sec:appendix-definitions}.}
  \label{fig:appendix-robustness}
\end{figure*}

\subsection{Comparison of tail-edge definitions}
\label{sec:appendix-edge}

The tail edge admits several reasonable definitions, differing in
whether the tail is measured by its $L_1$ mass,
$M_{\rm abs}(t,w) = \sum_{w'\ge w}|R_0(w',t)|$, or by its net signed
contribution $H(t,w)$ [Eq.~\eqref{eq:S}], and---since $S$ is not
monotone in $w$---whether the edge is the first weight at which the
quantity dips below tolerance (``literal'') or the smallest weight
beyond which it remains below tolerance (``contiguous'',
Eq.~\eqref{eq:wstar}). For the monotone $L_1$ mass the two criteria
coincide, leaving three distinct edges. At equal tolerance
$\epsilon_{\rm tail}$, they are nested at every time,
\begin{equation} \label{eq:nesting}
  w^{*}_{\rm abs}(t) \;\ge\; w^{*}_{\rm con}(t) \;\ge\;
  w^{*}_{\rm lit}(t),
\end{equation}
since $H(t,w) \le M_{\rm abs}(t,w)$ pointwise and the contiguous
criterion is stricter than the literal one; we verify the nesting
numerically at every readout time.
Figure~\ref{fig:appendix-edge-O-tfim} compares the three edges for
the two main-text TFIM states. As in
Fig.~\ref{fig:appendix-robustness}, the edges are computed without
the jitter median, and additionally without the time window of
Sec.~\ref{sec:appendix-definitions}, so that the bare per-time
behavior is visible.

The edges agree through the rise and most of the plateau, and differ
through two cancellation mechanisms. First, sign cancellations pull
the signed-sum edges below the $L_1$ edge: the literal edge dips
whenever the net tail happens to cross tolerance at an isolated
weight, and vanishes at zero-crossings of $\langle O(t)\rangle$
itself [where $H(t,0)=|\langle O(t)\rangle|$ forces
$w^{*}_{\rm lit}=0$]; the contiguous edge is affected only when the
entire upper tail net-cancels, producing sporadic transient dips.
Second, after $\langle O(t)\rangle$ has decayed, the two signed-sum
edges fall to zero---no net high-weight contribution remains---up to
the cancellation transients discussed in
Sec.~\ref{sec:appendix-definitions}, whereas the $L_1$ edge stays
finite, reflecting operator structure that is present but no longer
contributes.

We discard the literal edge as an unstable summary, and adopt the
contiguous signed-sum definition over the $L_1$ definition for a
reason of experimental accessibility: Pauli path spectroscopy
estimates smooth functionals of the signed reactivity,
$\sum_w h(w) R_\tau(w)$, so the signed tail is the quantity the
protocol can directly access, whereas the $L_1$ mass is not. Where
$\langle O(t)\rangle$ is resolvable above the tail tolerance, the
two agree up to sporadic cancellation dips of the contiguous edge.

\begin{figure}[t]
  \centering
  \includegraphics[width=\textwidth]{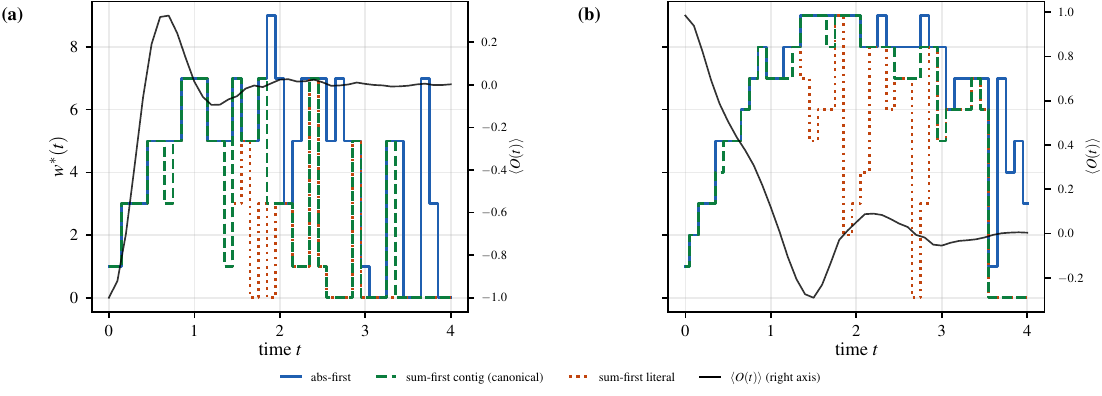}
    \caption{\textbf{Tail-edge definitions vs.\ $\langle O(t)\rangle$ (TFIM).}
    The three edges of Eq.~\eqref{eq:nesting} (left axis)---the $L_1$
    edge $w^{*}_{\rm abs}$ (solid), the contiguous signed edge
    $w^{*}_{\rm con}$ (dashed; the definition used in the main text),
    and the literal signed edge $w^{*}_{\rm lit}$ (dotted)---overlaid
    on $\langle O(t)\rangle$ (right axis, black), for \textbf{(a)}~the
    tilted state $\theta = -\pi/18$ and \textbf{(b)}~the domain-wall
    state. All edges are computed at
    $\epsilon_{\rm tail} = 10^{-2}$ without the jitter median or
    time window of Sec.~\ref{sec:appendix-definitions}, so that the
    bare per-time behavior is visible. The nesting
    $w^{*}_{\rm abs} \geq w^{*}_{\rm con} \geq w^{*}_{\rm lit}$ holds
    at every readout time.     Where $\langle O(t)\rangle$ is resolvable, the three edges largely agree; the signed edges dip where the upper tail net-cancels, and $w^{*}_{\rm lit}$ additionally vanishes at zero-crossings of $\langle O(t)\rangle$. In the decayed regime (shaded), the signed edges fall to zero while $w^{*}_{\rm abs}$ remains finite.}
  \label{fig:appendix-edge-O-tfim}
\end{figure}

\subsection{Extended state dependence}
\label{sec:appendix-states}
 
This subsection extends the diagnostics of the main text to the full
state sets of Sec.~\ref{sec:appendix-models-settings}.
 
\paragraph{Convergence.}
Figure~\ref{fig:appendix-expval} shows the evolved expectation value
$\langle O(t)\rangle$ for all initial states of both models.
The MFIM curves are essentially threshold-independent over the full
time window, consistent with the rapid convergence of LI methods for
this model. The TFIM curves display a visible threshold dependence
at intermediate times for several states, signalling the onset of
the costly regime; the sensitivity of the diagnostics themselves to
these thresholds is quantified in
Sec.~\ref{sec:appendix-robustness}.

\begin{figure*}[t]
  \centering
  \includegraphics[width=1.\textwidth]{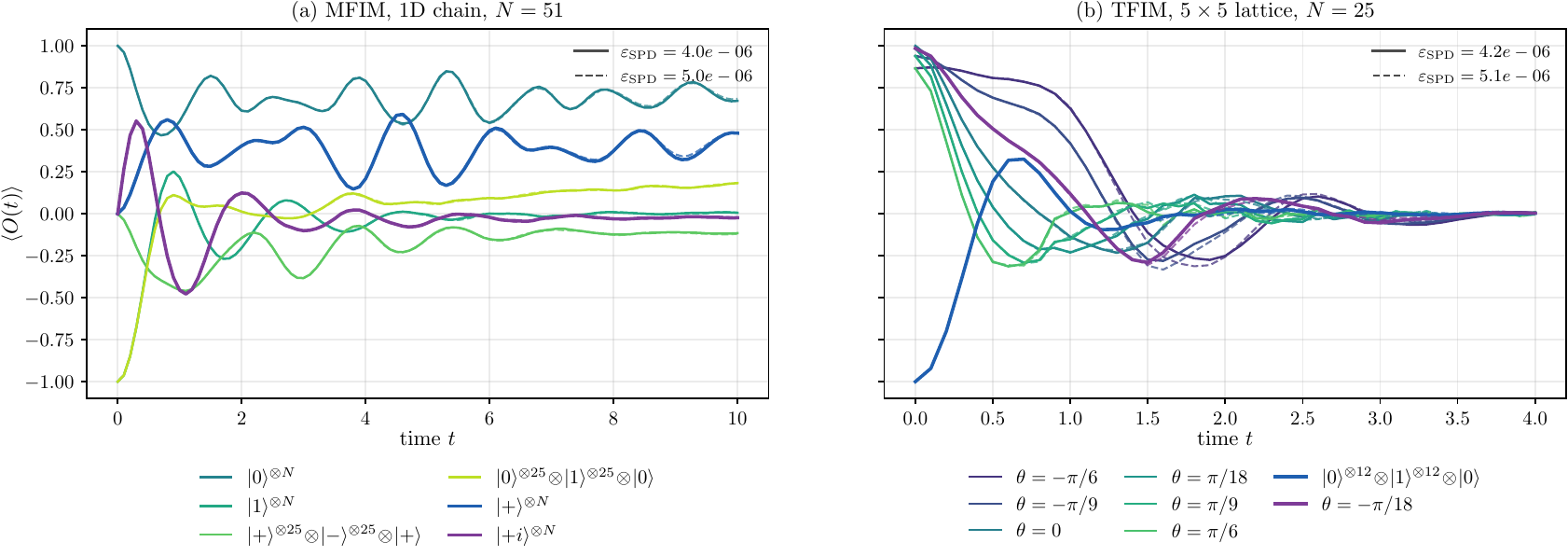}
  \caption{\textbf{Convergence, MFIM and TFIM.} (a)) Evolved expectation value
    $\langle O(t)\rangle$ for the six MFIM initial states, at two SPD
    truncation thresholds $\epsilon$ (solid/dashed). (b) Same for the nine TFIM initial
    states.}
  \label{fig:appendix-expval}
\end{figure*}

\paragraph{Per-state diagnostics.}
Figures~\ref{fig:appendix-mfim-grid} and~\ref{fig:appendix-tfim-grid}
show the tail edge $w^{*}(t)$ (left axis) and memory cost
$n_{\rm Pauli}(t)$ (right axis, logarithmic scale) for every state
of both models; shading marks the decayed regime. Across the full
sets, the two diagnostics track one another remarkably closely,
rising, plateauing, and decaying together while the observable is
resolvable. The largest deviations occur in the MFIM, for
$|0\rangle^{\otimes N}$ and in particular $|1\rangle^{\otimes N}$,
where $n_{\rm Pauli}(t)$ falls or fluctuates strongly at late times
while the tail edge remains large. We note that absolute values of
$n_{\rm Pauli}$ are not comparable between the two models, which are
simulated at different Trotter steps
(Sec.~\ref{sec:appendix-models-settings}); the tail edge, being
essentially independent of the discretization
(Sec.~\ref{sec:appendix-robustness}), may be compared directly.

In the TFIM, the diagnostics are strikingly uniform across the
tilted family: every tilt reaches a comparable plateau,
$w^{*} \approx 9$, at comparable times. Since plateau values
approaching the largest weights resolved by SPD are lower bounds
(Sec.~\ref{sec:appendix-robustness}), residual differences between
the tilts may be unresolved. The $Z$-basis domain-wall state, by
contrast, is comparatively inexpensive in both models, in keeping
with the early decay of its expectation value. More broadly, since
all states of a given model probe the identical Heisenberg-evolved
operator, the order-of-magnitude spread of $n_{\rm Pauli}$ across
states and times demonstrates that the simulation cost of the SPD algorithm is a property of the fully quantum
experiment, including the state, dynamics, and observable jointly, rather than
of the time-evolved operator alone.

\begin{figure*}[t]
  \centering
  \includegraphics[width=0.85\textwidth]{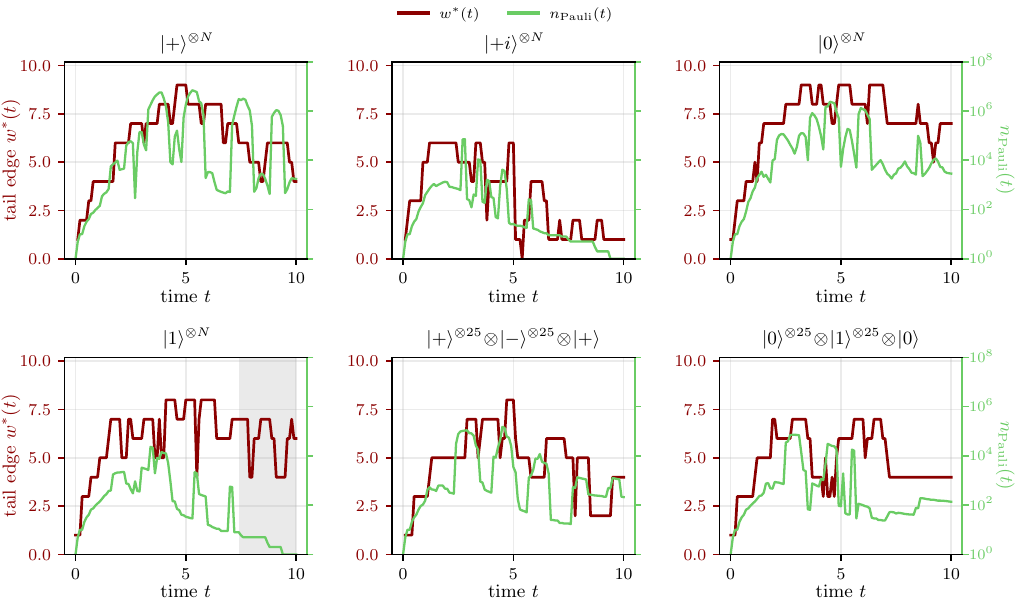}
\caption{\textbf{Extended state dependence, MFIM.} Tail edge
    $w^{*}(t)$ (solid, left axis) and memory cost $n_{\rm Pauli}(t)$
    (dashed, right axis, log scale) for each of the six MFIM states
    of Sec.~\ref{sec:appendix-models-settings}; shading marks the
    decayed regime of Sec.~\ref{sec:appendix-definitions}.}
  \label{fig:appendix-mfim-grid}
\end{figure*}

\begin{figure*}[t]
  \centering
  \includegraphics[width=0.85\textwidth]{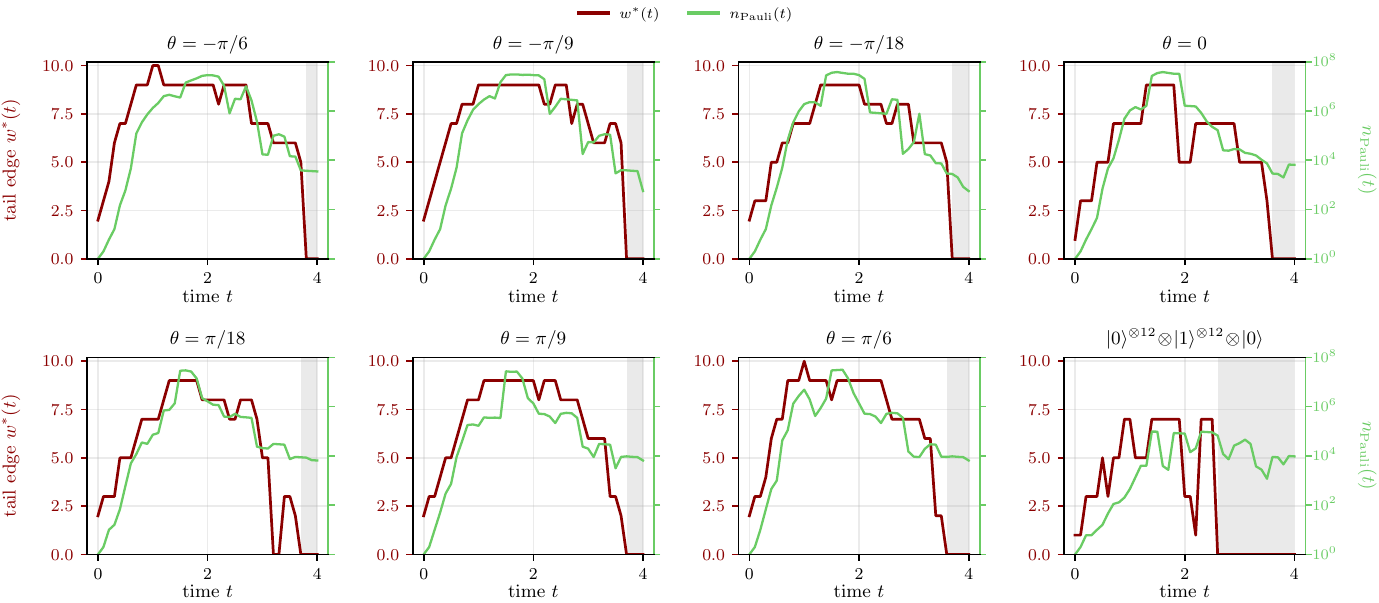}
\caption{\textbf{Extended state dependence, TFIM.} As in
    Fig.~\ref{fig:appendix-mfim-grid}, for the eight TFIM states.
    Panels $k = 0$--$6$ show the tilted family
    $\theta_k = (k-3)\,\pi/18$; the last panel is the $Z$-basis
    domain-wall state (Sec.~\ref{sec:appendix-models-settings}).}
  \label{fig:appendix-tfim-grid}
\end{figure*}

\section{Pauli path spectroscopy}

In this section, we provide full details and derivations for each of our protocols to experimentally measure the reactivity.
We begin, and spend the majority of time on, our ``decoherent'' protocol [Fig.~\ref{fig:3}(a)] involving inserted noise.
We provide an extensive discussion of numerical and analytical methods to optimize the protocol's performance and additional numerical studies on reactivity functions encountered in practice.
We then turn to our coherent measurement protocol, and provide a careful derivation relating the protocol's outcomes and the Fourier transform of the reactivity.

\subsection{Decoherent measurement protocol}

We now provide the details of our decoherent protocol to measure the reactivity.
The description of the protocol is straightforward, and is provided in the first two subsections.
The remaining subsections are devoted to the choice of protocol parameters and analysis of the experimental results.

\subsubsection{Description of the decoherent protocol} \label{sec: description decoherent}

We consider estimating the reactivity of any quantum experiment, involving a state $\psi$, unitary $U$, and observable $O$.
Our protocol applies equally to the space-time reactivity and the reactivity at a fixed circuit layer $t$.
For concreteness, we focus our presentation on the reactivity at a fixed circuit layer.

Our protocol estimates the reactivity by introducing artificial noise events into the experiment.
As discussed in Section~\ref{sec: reactivity}, in the absence of any such events, the experimental outcome can be expanded as,
\begin{equation}
    \tr(O U \psi U^\dagger ) = \sum_w R_\tau(w) = \sum_w \tr( U^\dagger_{t-\tau} O U_{t-\tau} \mathcal{P}_w \big[ U_\tau \psi U^\dagger_\tau \big] ),
\end{equation}
where $R_\tau(w)$ captures the contribution of all Pauli paths with weight $w$ at time $\tau$, and we expand $U = U_{t-\tau} U_\tau$.
Now, we consider performing the same experiment, but with single-qubit depolarizing noise of strength $\gamma$ artificially applied to each qubit in between circuit layers $\tau$ and $\tau+1$.
That is, we apply the channel, $\mathcal{D}_\gamma \equiv \bigotimes_{i=1}^N \mathcal{D}_{i,\gamma}$, with $\mathcal{D}_{i,\gamma}[\rho] \equiv e^{-\gamma} \rho + (1-e^{-\gamma}) \tr_i(\rho) \cdot \identity_i/2$.
The channel acts on Pauli operators as $\mathcal{D}_\gamma[P] = e^{-\gamma w[P]} P$.
Given this simple action, the noisy experimental outcome is equal to, 
\begin{equation}
    C_\tau(\gamma) \equiv \tr( U^\dagger_{t-\tau} O U_{t-\tau} \mathcal{D}_\gamma \big[ U_\tau \psi U^\dagger_\tau \big] ) = \sum_w R_\tau(w) e^{-\gamma w}.
\end{equation}
The contribution of each weight is damped by the exponential $e^{-\gamma w}$.
The implementation of artificial noise could be achieved, for example, by performing a random single-qubit Pauli rotation, $P_i \in \{ X_i, Y_i, Z_i \}$, with probability $p_\gamma = (1-e^{-\gamma})(3/4)$, on each qubit $i=1,\ldots,N$.

In principle, one can repeat the above protocol for $N$ distinct values $\{ \gamma_i \}$ and take a linear combination of the results to recover the exact reactivity function.
That is, one estimates the reactivity as
\begin{equation}
    \hat{R}_\tau(w) \equiv \sum_\gamma b_{w \gamma} \hat C_\tau(\gamma),
\end{equation}
where the coefficients $b_{w \gamma}$ are the matrix inverse of $a_{\gamma w'} \equiv e^{-\gamma w'}$ (i.e.~$b_{w \gamma} = (\bs{a}^{-1})_{w \gamma}$ where $\bs{a}$ denotes the matrix with elements $a_{\gamma w'})$, and $\hat C_\tau(\gamma)$ is the experimental estimate of $C_\tau(\gamma)$.
With this choice of $b_{w \gamma}$, one can easily verify that the experimental estimate, $\hat{R}_\tau(w)$, converges to the true reactivity in expectation.
This inversion is known as the discrete inverse Laplace transform (also known as the ``inverse Z-transform'' when the $\gamma$ are equally spaced).

Unfortunately, the uncertainties produced in the inverse Laplace transform are typically extremely large, causing the approach to fail for all but the smallest system sizes.
This is caused by ill-conditioning in the matrix $b_{w \gamma} = (\bs{a}^{-1})_{w \gamma}$, which contains both small singular values, of order one, and extremely large singular values, growing exponentially in $N$.
The uncertainty in our estimate $\hat{R}_\tau(w)$ is quantified by the covariance matrix,
\begin{equation}
    \Sigma_{w w'} \equiv \text{Cov}( \hat{R}_\tau(w) , \hat{R}_\tau(w') ) = \sum_{\gamma \gamma'} b_{w \gamma} \Sigma^0_{\gamma \gamma'} b_{w' \gamma' }, \,\,\,\,\,\, \,\,\,\,\,\, \Sigma^0_{\gamma \gamma'} = \delta_{\gamma \gamma'} \varepsilon_\gamma^2,
\end{equation}
where $\varepsilon_{\gamma}$ is the uncertainty in the experimental estimate of each $C_\tau(\gamma)$.
For example, if the uncertainty arises from shot noise in the measurement outcomes, one has $\varepsilon_{\gamma}^2 \sim \mathcal{O}(1/M_{\gamma})$ where $M_{\gamma}$ is the number of shots for each $\gamma$.
When the matrix $\bs{b}$ has large singular values, small uncertainties in the experimental measurements $C_\tau(\gamma)$ can translate to large uncertainties in the estimated reactivity.
Mathematical literature on the inverse Laplace transform suggests that these large singular values are fundamental to any noise-based approach~\cite{epstein2008bad}.

To sidestep this issue, we aim to only estimate \emph{low-uncertainty} features of the reactivity $R_\tau(w)$.
Roughly speaking, these correspond to features of $R_\tau(w)$ that overlap with the smaller singular values of $\bs{b}$.
To be specific, we consider the inner product of the reactivity with a ``filter function'', $h(w)$,
\begin{equation}
    \sum_w h(w) R_\tau(w).
\end{equation}
This can be estimated via
\begin{equation} \label{eq: estimate C from full}
    \sum_w h(w) \hat R_\tau(w) = \sum_\gamma h_\gamma \hat C_\tau(\gamma),
\end{equation}
where $h(w) = \sum_\gamma h_\gamma e^{-\gamma w}$,
with variance, $\text{Var}(\sum_w h(w) \hat R_\tau(w)) = \sum_\gamma h_\gamma^2 \varepsilon_\gamma^2$, where $\varepsilon_\gamma^2 \leq 1/M_\gamma$ is the uncertainty in the estimate of each $C_\tau(\gamma)$.
We provide a detailed analysis of the behavior of the variance for different filter functions in Sections~\ref{sec: numerical filter functions} and~\ref{sec: smoothed reactivity}.

As discussed in the main text and End Matter, we find that, in general, the variance above is small when $h(w)$ is \emph{smooth} as a function of the weight $w$. 
This follows because the ill-conditioning of the inverse Laplace transform arises mainly when estimating high-frequency components of $R_\tau(w)$, i.e.~those components that oscillate quickly in $w$.
As discussed in the End Matter, this motivates us to define two smoothed variants of the reactivity function.
First, when $h^{\text{delta}}(w;w_c)$ is chosen to approximate a delta function about $w \approx w_c$, we define the \emph{smoothed reactivity function},
\begin{equation}
    \tilde{R}_\tau(w_c) \equiv \sum_{w} h^{\text{delta}}(w;w_c) R_\tau(w),
\end{equation}
where $h^{\text{delta}}(w;w_c)$ will typically have width $\delta w = \mathcal{O}(w_c)$ to maintain sample efficiency.
Second, when 
$h^{\text{Heaviside}}(w;w_c)$ is chosen to approximate a Heaviside function at $w \approx w_c$, we define the \emph{smoothed cumulative reactivity function},
\begin{equation}
    \tilde{H}_\tau(w_c) \equiv \sum_{w} h^{\text{Heaviside}}(w;w_c) R_\tau(w),
\end{equation}
where $h^{\text{Heaviside}}(w;w_c)$ will typically rise from zero to one over width $\delta w = \mathcal{O}(w_c)$ to maintain sample efficiency.
In what follows, we often abbreviate $h^{\text{delta}}(w;w_c)$ or $h^{\text{Heaviside}}(w;w_c)$ as merely $h(w;w_c)$ or $h(w)$ when clear from context.

\subsubsection{Replacing noise events with random Pauli insertions} \label{sec: decoherent random Paulis}

In this section, we introduce a convenient second variant of our decoherent protocol.
This variant is relevant when the artificial noise channels in the protocol are implemented via random Pauli rotations.
When this is the case, it can be favorable to work directly with the Pauli rotations themselves instead of the noise channel $\mathcal{D}_\gamma$.
To this end, in this section we repeat the analysis of Section~\ref{sec: description decoherent} in this scenario.

Consider the same protocol as in Section~\ref{sec: description decoherent}, but where we replace the noise channel of strength $\gamma$ with the application of a random Pauli operator of weight $k$.
Averaging over all Paulis of a fixed weight $k$, this effectively performs the operation
\begin{equation}
    \mathcal{D}^{\text{Pauli}}_k[(\cdot)] = \frac{1}{3^k {N\choose k}} \sum_{Q : w[Q] = k} Q (\cdot) Q^\dagger,
\end{equation}
where $3^k {N\choose k}$ is the total number of Pauli operators of weight $k$.
To extract the reactivity from such operations, we must first derive the action of the channel $\mathcal{D}^{\text{Pauli}}_k$ on a Pauli operator $P$.
Since $\mathcal{D}^{\text{Pauli}}_k$ is invariant under permutations of the qubits as well as single-qubit Clifford rotations, this action depends solely on the weight of $P$.
A straightforward counting argument yields,
\begin{equation}
    \mathcal{D}^{\text{Pauli}}_k[P] 
    = \sum_{r=0}^{\min(k,w[P])} \left( - \frac{1}{3} \right)^r \frac{{w[P] \choose r} {N-w[P] \choose k-r}}{{N \choose k}}
     P
    = \prescript{}{2}F_1^{}(-w[P],-k,-N;4/3)  P,
\end{equation}
where $\prescript{}{2}F_1^{}(-w[P],-k,-N;4/3)$ is the hypergeometric function.
The first formula follows by counting the number of weight-$k$ Pauli operators $Q$ that overlap the weight-$w[P]$ Pauli operator $P$ on exactly $r$ qubits.
Each such $Q$ contributes, on average, a factor of $(-1/3)^r$, arising from the product of $r$ averages, $(X_i P_i X_i + Y_i P_i Y_i + Z_i P_i Z_i)/3 = -1/3$, on each qubit $i$ within their overlap.
The second equality follows by noting that the sum is equal to probability generating function of the hypergeometric distribution evaluated at $-1/3$ (see e.g. Eq.~(6.4) of Ref.~\cite{johnson2005univariate}).
In the limit $k \ll N$, the hypergeometric function simplifies, to $\prescript{}{2}F_1^{}(-w[P],-k,-N;4/3) \approx (1-4w/3N)^k$.
This corresponds to the damping when one replaces the random $k$-qubit Pauli operator with a product of $k$ independently random single-qubit Pauli operators.

With this formula in hand, the average outcome of an experiment with a random weight-$k$ Pauli operator inserted at layer $\tau$ is,
\begin{equation}
    C_\tau(k) \equiv \tr( U^\dagger_{t-\tau} O U_{t-\tau} \, \mathcal{D}_k^{\text{Pauli}} \big[ U_\tau \psi U^\dagger_\tau \big] ) = \sum_w R_\tau(w) \cdot \prescript{}{2}F_1^{}(-w,-k,-N;4/3).
\end{equation}
The exact reactivity can be estimated as,
\begin{equation}
    \hat{R}_\tau(w) \equiv \sum_k b^{\text{Pauli}}_{w k} C_\tau(k),
\end{equation}
where $b^{\text{Pauli}}_{wk}$ is the matrix inverse of $a^{\text{Pauli}}_{kw} = \prescript{}{2}F_1^{}(-w,-k,-N;4/3)$.
The uncertainty in the estimate is described by a covariance matrix,
\begin{equation}
    \Sigma_{w w'} \equiv \text{Cov}( \hat{R}_\tau(w) , \hat{R}_\tau(w') ) = \sum_{k k'} b^{\text{Pauli}}_{w k} \, \Sigma^0_{k k'} \, b^{\text{Pauli}}_{w' k' }, \,\,\,\,\,\, \,\,\,\,\,\, \Sigma^0_{k k'} = \delta_{k k'} \varepsilon_k^2,
\end{equation}
where $\varepsilon_k$ is the uncertainty in the experimental measurement of $C_\tau(k)$.
Similar to the first variant of the protocol, the covariance matrix will typically feature both very large and small singular values.
To this end, we can instead estimate the overlap of the reactivity with smoothed functions,
\begin{equation}
    h(w) = \sum_k h_k \cdot \prescript{}{2}F_1^{}(-w,-k,-N;4/3),
\end{equation}
similar to before.
Working with random Pauli insertions compared to local depolarizing noise yields significant advantages when the relevant weights $w$ are a significant fraction of the system size, e.g.~$w \gtrsim N/20$ in practice.

\subsubsection{Sampling overhead} \label{sec: numerical filter functions}

In this subsection, we address the sampling overhead of experimentally measuring the overlap of the reactivity function with a given filter function, $h(w)$.
This will also provide the optimal distribution of experimental shots to allocate to each noise rate $\gamma$ or number of random Pauli insertions $k$ for a given filter function of interest.
For concreteness, we focus on experiments with random Pauli insertions of weight $k$; all formulas immediately generalize to noise rates $\gamma$ as well.

As discussed, our estimate of $\sum_w h(w) R_\tau(w)$ is formed from a linear combination of  expectation values with inserted depolarizing noise, $\sum_\gamma h_\gamma \hat C_\tau(\gamma)$,
where $h(w) = \sum_\gamma h_\gamma e^{-\gamma w}$, or a linear combination of expectation value with random Pauli insertions, $\sum_k h_k \hat C_\tau(k)$,
where $h(w) = \sum_k h_k F(k,w;N)$.
We assume $M_k$ samples are allocated to each value of $k$, where the total number of samples is $\sum_k M_k = M$.
The uncertainty in each expectation value is given by the shot noise error,
\begin{equation} \label{eq: eps C M}
    \varepsilon_k^2 = \frac{(1+C_\tau(k)) ( 1 - C_\tau(k) )}{ M_k} \leq \frac{1}{ M_k}.
\end{equation}
The uncertainty in our estimate  is therefore,
\begin{equation} \label{eq: eps h eps gamma}
    \varepsilon_h^2 = \sum_k h_k^2 \varepsilon_k^2 \leq \sum_k \frac{h_k^2}{M_k}.
\end{equation}
We would like to minimize the upper bound over the choices of $\{ M_k \}$.
The constraint $\sum_k M_k = M$ can be handled by introducing a Lagrange multiplier.
A standard analysis yields the optimal choice,
\begin{equation}
    M_k^* = \left( \frac{ | h_k | }{\sum_k |h_k|} \right) M.
\end{equation}
This produces an uncertainty,
\begin{equation} \label{eq: uncertainty h}
    (\varepsilon_h^*)^2 \leq \frac{\left( \sum_k | h_k | \right)^2}{M} \equiv \frac{X}{M}.
\end{equation}
The numerator, $X\equiv ( \sum_k | h_k | )^2$, quantifies the increase in the number of samples required to estimate our feature $h(w)$ compared to a standard expectation value.
As expected, larger coefficients $h_k$ lead to a larger sampling overhead.

\subsubsection{Further details on numerical optimization of filter functions}

In this subsection, we provide further details on our numerical optimization procedure for obtaining filter functions with desirable properties at a moderate sampling overhead.
We refer to the End Matter for an introduction to this approach, as well as  our other, analytic constructions of filter functions.

We focus on random Pauli insertions for specificity, and all results immediately extend to noise rates $\gamma$.
The choice of the optimal filter function must balance two competing factors: the desire for $h(w)$ to closely resemble $f(w)$, and the desire for the uncertainty [Eq.~(\ref{eq: uncertainty h})] to be small.
We handle the first point by defining a cost function,
\begin{equation}
    \sum_w c(w) (h(w) - f(w))^2,
\end{equation}
which we seek to minimize.
The coefficients $c(w)$ can be chosen according to one's specific desiderata for the filter function $h(w)$.
For example, if $f(w)$ was a delta function, one might wish to increase $c(w)$ as one moves farther from the center of the delta function, in order to assure that $h(w)$ strongly vanishes in such regions.
We handle the second point by setting a maximum value, $(\sum_k |h_k|)^2 \leq X$, on the sampling overhead.
This reduces our selection of the filter function coefficients $h_k$ to a convex optimization problem that can be easily solved numerically.
The resulting filter functions produced by this optimization are shown in our figures throughout the main text, End Matter, and Supplemental Material.

We find that generically, the sampling overhead $X$ and the achievable inverse relative width $w_*/\delta w$ (where the width $\delta w$ is, roughly speaking, the length-scale that the filter function varies on near weight $w_*$) of the numerically optimized filter functions are related as $X \sim e^{\mathcal{O}(w_*/\delta w)}$ (see Fig.~\ref{fig:filter-comparison}).
This matches the scaling predicted analytically from our Fourier-to-Chebyshev approach (see End Matter).
From this, we conclude that such a sampling overhead is likely fundamental to any estimation strategy using the decoherent measurement protocol.

\begin{figure*}[t]
  \centering
  \includegraphics[width=\textwidth]{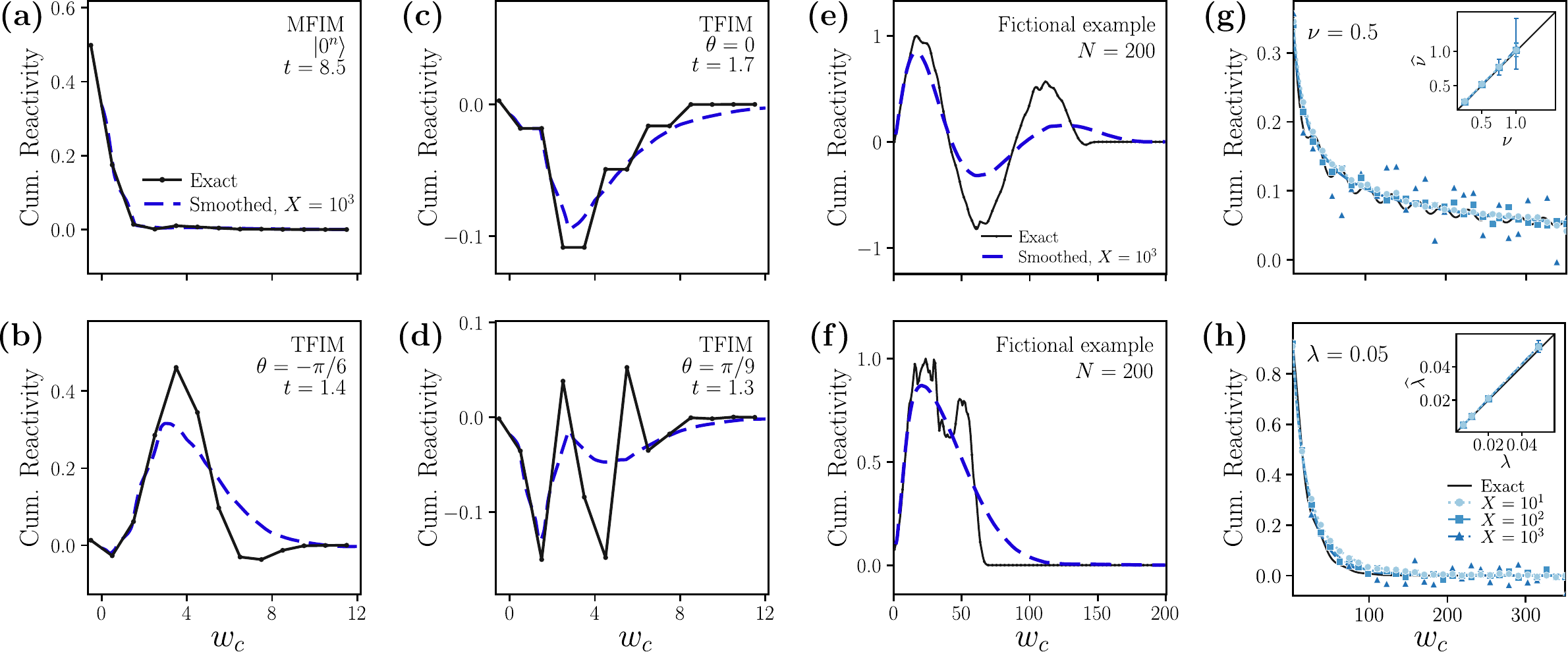}
  \caption{Comparisons of exact cumulative reactivity functions, $H(w_*)=\sum_{w>w_*}R(w)$ (black), and smoothed cumulative reactivity functions $\widetilde H(w_*)=\sum_w h(w;w_*)R(w)$, (blue) that can be efficiently measured, with sampling overhead $X=10^3$.
  \textbf{(a-d)} Examples from SPD numerical simulations for the 1D MFIM with initial state $|0^n\rangle$ at $t=8.5$, and the 2D TFIM with $\theta=-\pi/6$ at $t=1.4$, $\theta=0$ at $t=1.7$, and $\theta=\pi/9$ at $t=1.3$, respectively. 
  \textbf{(e-f)} Analogous comparisons for two fictional exact reactivity functions, to illustrate smoothing behavior at large $w_c$.
  \textbf{(g-h)} Two further examples of fictional reactivity functions, obeying a power-law and exponential decay with $w_c$, respectively.
  Insets denote the power-law degree and exponential decay rate estimated from smoothed reactivity data, which agree nearly precisely with their exact values. We refer to the text for further discussion.
}
  \label{fig:smoothed_examples}
\end{figure*}

\begin{figure*}[t]
  \centering
  \includegraphics[width=\textwidth]{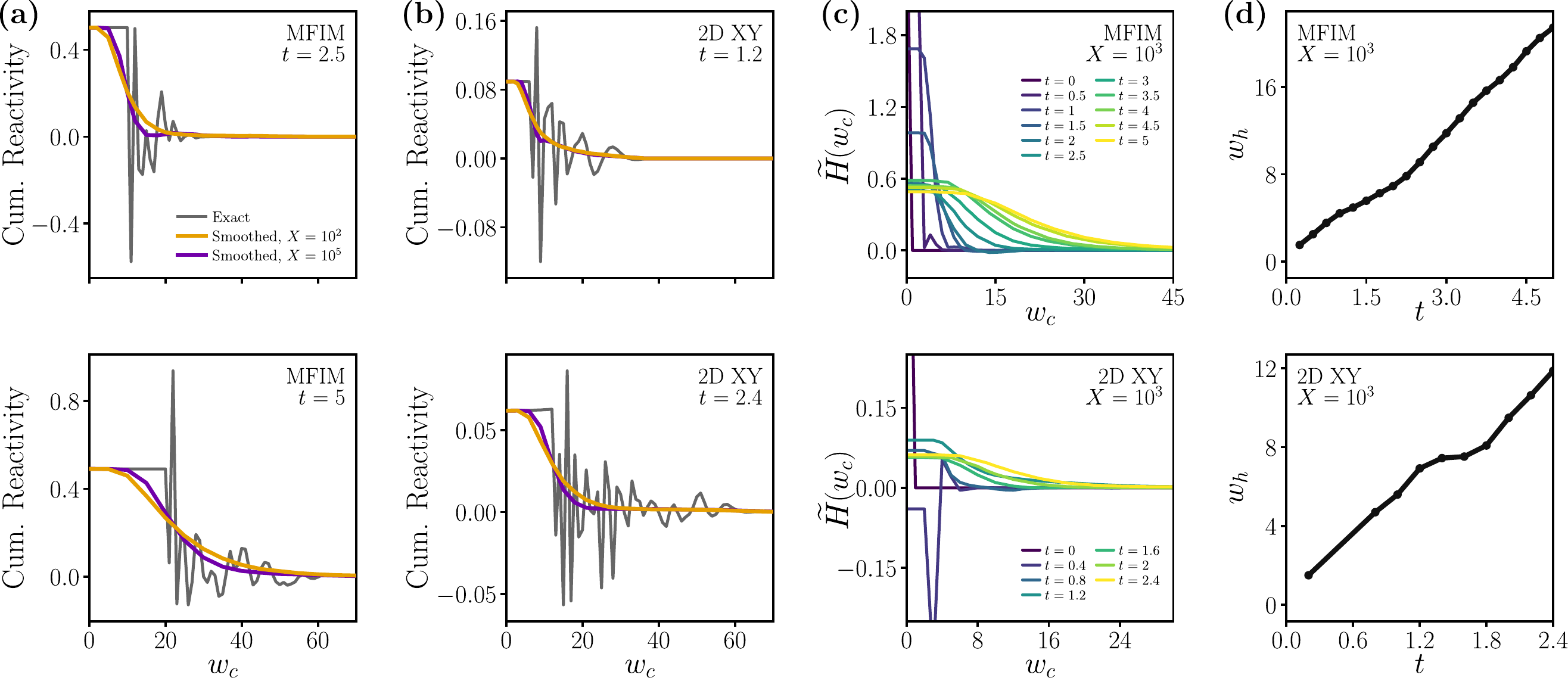}
  \caption{Numerical simulations of Pauli path spectroscopy for the global reactivity, which allows us to demonstrate the performance of the protocol at large weights.
  We consider the infinite-temperature autocorrelation function, $\frac{1}{2^N}\tr(h_1(t) h_1(0))$, of the local energy density $h_1$ in the 1D MFIM on $N=20$ qubits, and for a similar function, $\frac{1}{2^N}\tr(Z_{3,3}(t) Z_{3,3}(0))$, in the 2D XY model on a $4\times4$ lattice. 
  We trotterize each time-evolution with step size $\delta t = 0.25, 0.05$, respectively, use coefficient thresholds for the SPD algorithm of $\epsilon_{\text{SPD}} = 10^{-4}, 6 \times 10^{-4}$, respectively, and let $w[\vec{P}] = \sum_{\tau=1}^T w[P_\tau]$, where $\tau$ is sampled every $1, 4$ trotter steps, respectively.
  \textbf{(a,b)} Grey curves show the exact cumulative reactivity \(H(w_c)=\sum_{w>w_c}R(w)\), and colored curves show the smoothed cumulative reactivity \(\tilde{H}(w_c)\) obtained from random Pauli insertions for sampling overheads \(X=10^2\) and \(10^5\).
  \textbf{(c)}~Smoothed cumulative reactivity functions over time in the MFIM (top) and XY model (bottom), at sampling overhead \(X=10^3\).
  \textbf{(d)}~The ``halfway weight'' \(w_h\), at which \(\tilde{H}(w_c)\) falls below half its initial value, $\tilde{H}(w_h) \leq \frac{1}{2} \tilde{H}(0)$.
  Data points at early times in which $\tilde{H}(w_c)$ is not monotonically decreasing with $w_c$ are omitted.
  }
  \label{fig:global_supp}
\end{figure*}

\subsubsection{Numerical studies of smoothed reactivity functions} \label{sec: smoothed reactivity}

In this subsection, we present additional numerical studies of smoothed reactivity functions, complementing those in Fig.~\ref{fig:3} of the main text.

We begin with Fig.~\ref{fig:smoothed_examples}.
In Fig.~\ref{fig:smoothed_examples}(a-d), we show additional comparisons between smoothed cumulative reactivity function and exact cumulative reactivity functions for the same models studied using the SPD algorithm in the main text.
For extremely small $w_c \lesssim 3$, the smoothed filter function is able to exactly approximate the Heaviside function, and so the smoothed and exact reactivity functions agree.
For any larger $w_c$, they do not necessarily agree, and in general, across all four panels, the smoothed reactivity function resembles a smooth interpolation of the exact reactivity function, as anticipated.
As discussed in the main text, this still enables access to many interesting features of the reactivity, such as its typical weight over time.
It is also in some sense the only ``physical component'' of the exact reactivity function, as any sharp oscillations in the reactivity are not observable in quantum experiments (although, they could be observable and important for classical algorithms that perform weight-based truncation).

To illustrate the behavior the smoothed filter functions at weights $w_c$ beyond the reach of our numerical simulations---i.e.~where one hopes to apply them in  beyond-classical quantum experiments---in Fig.~\ref{fig:smoothed_examples}(e-f), we create fictional examples of reactivity functions and compare their exact versus smoothed cumulative reactivities.
We observe identical behavior at these large scales as already observed at small scales.
Namely, that the smoothed reactivity resembles a smooth interpolation of the exact reactivity as anticipated.
In particular, we see that the smoothed reactivity is still capable of resolving the difference in typical weight between (e) and (f), and of resolving the large oscillations in (e).

Finally, in Fig.~\ref{fig:smoothed_examples}(g-h), we display an example of how the smoothed cumulative reactivity can enable one to accurately extract simple parameters governing the exact cumulative reactivity.
In Fig.~\ref{fig:smoothed_examples}(g), we devise a fictional exact cumulative reactivity function that decays as a power law in $w_c$, $H(w_c) = [1+0.1 \sin(w_c/5)]/w_c^\nu$, with small oscillations added for generality.
We then ask whether one can extract the power law degree $\nu$ solely from the smoothed reactivity data.
We find that performing a best fit to the functional form $A/w_c^\nu$, with $A$ and $\nu$ as fitting parameters, yields an accurate estimate $\hat{\nu}$ of $\nu$ at all values tested.
In Fig.~\ref{fig:smoothed_examples}(h), we perform a similar test for an exact cumulative reactivity function that decays exponentially, as $H(w_c) = [1+0.1 \sin(w_c/5)]e^{-\lambda w_c}$.
In this case, we find that performing a functional fit that is aware of the smoothing effect of the filter function yields best results.
Namely, we propagate each trial exponential decay function \(\hat{A} e^{-\hat{\lambda} w_c}\) through our filter functions before fitting.
This yields an accurate estimate $\hat{\lambda}$ of $\lambda$ at all values tested.

We now turn to our second additional numerical study, in  Fig.~\ref{fig:global_supp}.
Here, we perform numerical computations of the \emph{global} reactivity function using the SPD algorithm.
As discussed in the main text, the global reactivity counts the total contribution of all Pauli paths of weight $w$, where the weight is summed over all circuit locations rather than at only a single instant in time.
By taking the time to be large, the global reactivity can acquire support at large weights $w$ even in a small-size system.

In Fig.~\ref{fig:global_supp}(a) and Fig.~\ref{fig:global_supp}(b), we display the exact and smoothed cumulative global reactivity functions for the 1D MFIM and the 2D XY model, respectively.
In both cases, we observe that the exact global reactivity functions features sharp oscillations as a function of $w_c$, which will necessarily be averaged over in any noisy expectation value.
The smoothed reactivity averages over these oscillations, and so features a gradual decay from its initial value, at small $w_c$, to zero, over a range of weights that mirrors the support of the exact reactivity.
To probe the evolution of this support in time, in Fig.~\ref{fig:global_supp}(c), we plot the smoothed cumulative reactivity function for each model as a function of $w$ for various times $t$.
We observe a gradual growth of the distribution towards higher weights over time.
To better visualize this growth, in Fig.~\ref{fig:global_supp}(d), we plot the weight $w_h$ at which the smoothed cumulative reactivity decays to half its initial value as a function of time for each model.
Interestingly, for both models, we observe that this weight grows roughly linearly in time.
This contrasts with the quadratic or cubic growth of the light-cone volume in 1D and 2D, respectively; providing evidence that the Pauli paths relevant to quantum expectation values can be of parametrically smaller weights than those relevant to exact operator evolution.

\subsection{Coherent measurement protocol}

We now provide the details of our coherent protocol to measure the reactivity.
Compared to the decoherent measurement protocol, the coherent protocol has the primary benefit that it can extract the precise values of the reactivity function with minimal sampling overhead.
Its primary drawbacks are that it requires an extra half-circuit of evolution time, and is only applicable to infinite-temperature correlation functions.
It also behaves substantially differently under experimental noise compared to the decoherent measurement protocol.

We discuss two variants of the coherent measurement protocol.
The first \emph{ancilla-assisted} variant uses $N$ ancillas prepared in a tensor product of EPR pairs with the system of interest.
This allows one to extract the reactivity function, $R_\tau(w)$, from the experimental results via a simple inverse Fourier transform.
The second \emph{ancilla-free} variant uses an extremely similar method but avoids the need for ancilla qubits.
This variant allows one to extract a ``classical'' variant of the reactivity function, $R_\tau(w_c)$, again via a simple inverse Fourier transform.
Here, the ``classical'' weight $w_c$ counts the number of $X$ or $Y$ components, but not $Z$ components, of a Pauli operator.

Our protocols are inspired by the size distribution measurement protocol introduced in Ref.~\cite{schuster2022many}. 

\subsubsection{Ancilla-assisted  protocol} \label{sec: coherent ancilla-assisted}

We begin with the ancilla-assisted protocol, depicted in Fig.~\ref{fig:coherent}.
We consider measuring the reactivity of an infinite-temperature correlation function, $\tr( B U A U^\dagger )/2^n$, where $A$ and $B$ are Pauli operators.
For a given time slice $t$, the reactivity is given by,
\begin{equation}
    R_\tau(w) = \tr( U^\dagger_{t-\tau} B U^{}_b \mathcal{P}_w \big[  U^{}_a A U_\tau^\dagger  \big] ) / 2^n,
\end{equation}
where we abbreviate $U = U_{t-\tau} U_\tau$, with $U_\tau = U_\tau \cdots U_1$ representing time-evolution up to layer $t$, and $U_{t-\tau} = U_t \cdots U_{\tau+1}$ after.
Traditionally, the infinite-temperature correlation function would be measured by preparing an initial state, $\psi_A = (\identity+A)/2^n$,
performing $U$, and measuring $B$.
For example, if $A$ is a single-qubit Pauli operator on qubit $i$, the state $\psi_A$ is the positive eigenstate of $A$ on qubit $i$ and maximally mixed on all other qubits.
In what follows, we will take a quite different approach in order to coherently measure the reactivity.

Our approach utilizes the Choi–Jamiolkowski mapping from operators $A$ on a Hilbert space $\mathcal{H}$, to states,
\begin{equation}
    \ket{A} \equiv ( A \otimes \identity )\ket{\text{EPR}},
\end{equation}
on a doubled Hilbert space $\mathcal{H} \otimes \mathcal{H}$.
Here, $\ket{\text{EPR}} \equiv \bigotimes_{i=1}^N \big( (\ket{00}_i + \ket{11}_i)/\sqrt{2} \big)$ is a tensor product of $N$ EPR pairs between the $N$ system qubits and $N$ ancilla qubits.
The un-perturbed EPR state corresponds to the identity operator $\ket{\identity} = \ket{\text{EPR}}$.
The benefit of this approach is that weight of an operator on $\mathcal{H}$ becomes a Hermitian observable on the doubled system.
Namely, we can define the weight operator,
\begin{equation} \label{eq: weight operator}
    W = - \frac{1}{4} \sum_i \left( X_i \otimes X_i - Y_i \otimes Y_i + Z_i \otimes Z_i - 3 \identity_i \otimes \identity_i \right),
\end{equation}
acting on $\mathcal{H} \otimes \mathcal{H}$.
The weight operator is a sum of local two-qubit terms, which couple a system qubit $i$ with its partner ancilla qubit.
The eigenstates of the weight operator are given by the Pauli-perturbed states,
$W ( P \otimes \identity ) \ket{\text{EPR}} = w[P] ( P \otimes \identity ) \ket{\text{EPR}}$, with eigenvalue equal to the Pauli operator's weight.
We note that the final term in Eq.~(\ref{eq: weight operator}) is an overall constant that can be neglected in practice.

The fact that the weight operator is a sum of local couplings opens the door to a simple protocol to measure the reactivity [Fig.~\ref{fig:coherent}].
We prepare the $N$ system and $N$ ancilla qubits in the EPR state.
We then perturb the EPR state with a local $\frac{\pi}{4}$-rotation, $(\identity \pm i A)/\sqrt{2}$, on the system side.
We then ``time-evolve'' this perturbation, by applying the unitary $U_\tau$ to the system side and $U_\tau^*$ to the ancilla side.
(Alternately, we could have applied the unitary $U_\tau^\dagger$ to the system side prior to the $\frac{\pi}{4}$-rotation, instead of $U_\tau^*$ after, using the identity $(\identity \otimes U_\tau^*)\ket{\text{EPR}} = (U_\tau^\dagger \otimes \identity )\ket{\text{EPR}}$.)
We then couple the two sides by evolving under the weight operator, as $e^{i g W}$, for a rotation angle $g$.
Finally, we apply the unitary $U_{t-\tau}$ to the system side and conclude by measuring the expectation value of the observable $B$ (on the system side).
In total, this produces a experimental outcome,
\begin{equation}
    \langle B \rangle_\pm = \frac{1}{2} \bra{\text{EPR}} \big( U_\tau [\identity \mp i A] U_\tau^\dagger \otimes \identity \big) e^{-ig W}  \big( U_{t-\tau}^\dagger B U_{t-\tau} \otimes \identity \big)  e^{ig W} \big( U_\tau [\identity \pm i A] U_\tau^\dagger \otimes \identity \big) \ket{\text{EPR}},
\end{equation}
which depends on the direction $\pm$ of the $\frac{\pi}{4}$-rotation, as well as the tunable angle $g$.

To extract the reactivity from these experimental outcomes, we evaluate the difference between the expectation value for the positive and negative $\frac{\pi}{4}$-rotations,
\begin{equation}
\begin{split}
    \mathcal{F}_t(g) \equiv \langle B \rangle_+ - \langle B \rangle_-  
    & = i \bra{\text{EPR}} e^{-ig W}  \big( U_{t-\tau}^\dagger B U_{t-\tau} \otimes \identity \big)  e^{ig W} \big( U_\tau A U_\tau^\dagger \otimes \identity \big) \ket{\text{EPR}} + h.c. \\
    & = i \bra{\text{EPR}} \big( U_{t-\tau}^\dagger B U_{t-\tau} \otimes \identity \big)  e^{ig W} \big( U_\tau A U_\tau^\dagger \otimes \identity \big) \ket{\text{EPR}} + h.c. \\
    & = i \tr( U_{t-\tau}^\dagger B U_{t-\tau} \cdot e^{ig \mathcal{W}} [ U_\tau A U_\tau^\dagger ] ) / 2^n + h.c. \\
    & = \sum_{w=1}^{N} 2\sin(gw) \cdot \tr( U_{t-\tau}^\dagger B U_{t-\tau} \cdot \mathcal{P}_w [ U_\tau A U_\tau^\dagger ] ) \\
    & = \sum_{w=1}^{N} 2\sin(gw) \cdot R_\tau(w), \\
\end{split}
\end{equation}
where in the third line, $\mathcal{W}[P] = w[P] P$, is the weight superoperator.
The result is equal to the Fourier transform of the reactivity.
To recover the reactivity from the experiment, we should repeat the protocol for $N$ different values of the rotation angle, $g \in \{ 2\pi/(2N+1), 4\pi/(2N+1), \ldots, 2N\pi/(2N+1) \}$.
This allows us to recover $\mathcal{F}_t( 2\pi k / (2N+1) )$ for all $k \in \{0,\ldots,2N\}$, using $\mathcal{F}_t(0) = 0$ and $\mathcal{F}_t(2\pi-g) = -\mathcal{F}_t(g)$.
We can then perform the inverse Fourier transform, to recover
\begin{equation}
    R_\tau(w) = \frac{i}{2N} \sum_{k=0}^{2N} e^{i 2\pi k w /(2N+1)} \mathcal{F}_t(2\pi k/(2N+1)) ,
\end{equation}
for each $w = 1,\ldots,N$.
This completes the measurement protocol.
The entire protocol requires $2N$ distinct experiments, and features a sampling error $\sim \mathcal{O}(1/\sqrt{M})$ in the estimate of each individual value of $R_\tau(w)$, where $M$ is total number of shots over all $2N$ experiments.

\subsubsection{Ancilla-free protocol} \label{sec: coherent ancilla-free}

We can also formulate an ancilla-free variant of the coherent measurement protocol.
This may be convenient in experiments with a limited total number of qubits, or whose connectivity does not easily allow the pairwise couplings needed for the ancilla-assisted protocol.

The ancilla-free coherent measurement protocol does not extract the reactivity function as discussed so far, but rather a closely-related variant of it.
We define the $Z$-weight of a Pauli operator $P$ to count its number of $X$ or $Y$ components, $w_Z[P] = \sum_{i=1}^N ( \delta_{P_i = X_i} + \delta_{P_i = Y_i} )$---i.e.~its number of components that anti-commute with the operator $Z$.
This contrasts with the weight, which also counts the $Z$ components.
Identical to before, we can define a reactivity function for the $Z$-weight.
For the infinite-temperature correlation function considered in this section, this is
\begin{equation}
    R^Z_t(w_Z) = \tr( U^\dagger_{t-\tau} B U^{}_b \mathcal{P}_{w_Z}^Z \big[  U^{}_a A U_\tau^\dagger  \big] ) / 2^n,
\end{equation}
where $\mathcal{P}_{w_Z}^Z[ \cdot ]$ projects onto Pauli operators with $Z$-weight $w_Z$.
We have $\sum_{w_Z=0}^N R^Z_t(w_Z) = \tr( B U A U^\dagger )$.
Note that the $Z$-weight may equal zero, corresponding to the contribution of Pauli operators entirely in the $Z$-basis.

The ancilla-free coherent measurement protocol proceeds as in Fig.~\ref{fig:3}.
We initialize the system in a random computational basis state, $\ket{\bs{s}}$, where $\bs{s} \in \{0,1\}^N$.
We then time-evolve backward under $U_\tau^\dagger$, apply the $\pm \frac{\pi}{4}$-rotation, $(\identity\pm A)/2$, and time-evolve forward under $U_\tau$.
We then apply the product of single-qubit rotations, $e^{ig W_Z(\bs{s})}$, for 
\begin{equation}
    W_Z(\bs{s}) = -\frac{1}{2} \sum_i \left( s_i Z_i - \identity_i \right).
\end{equation}
As before, the final term is a constant that can be neglected.
Finally, we time-evolve forward under $U_{t-\tau}$ and measure the observable $B$.

This protocol produces an expectation value,
\begin{equation}
    \langle B \rangle_{\pm,\bs{s}} = \frac{1}{2} \bra{\bs s}  U_\tau [\identity \mp i A] U_\tau^\dagger \cdot e^{-ig W_Z(\bs{s})} \cdot U_{t-\tau}^\dagger B U_{t-\tau} \cdot  e^{ig W_Z(\bs{s})} \cdot U_\tau [\identity \pm i A] U_\tau^\dagger  \ket{\bs s}.
\end{equation}
To extract the reactivity, we take the difference between the positive and negative $\frac{\pi}{4}$-rotations, and average over computational basis states,
\begin{equation}
\begin{split}
    \mathcal{F}_{t}^Z(g) & = \E_{\bs s} \big[ \langle B \rangle_{+,\bs{s}} - \langle B \rangle_{-,\bs{s}} \big] \\ 
    & = i \E_{\bs s} \big[ \bra{\bs s} e^{-ig W_Z(\bs{s})}   U_{t-\tau}^\dagger B U_{t-\tau} \cdot  e^{ig W_Z(\bs s)} \cdot U_\tau A U_\tau^\dagger   \ket{\bs s} \big] + h.c. \\
    & = i \E_{\bs s} \big[ \bra{\bs s}  U_{t-\tau}^\dagger B U_{t-\tau} \cdot  e^{ig W_Z(\bs s)} \cdot U_\tau A U_\tau^\dagger  \ket{\bs s} \big] + h.c. \\
    & = i \E_{\bs s} \left[  \tr(  U_{t-\tau}^\dagger B U_{t-\tau} \cdot  e^{ig W_Z(\bs s)} \cdot U_\tau A U_\tau^\dagger  \cdot \dyad{\bs s} ) \right] + h.c. \\
    & = i \E_{\bs s} \left[  \tr(  U_{t-\tau}^\dagger B U_{t-\tau} \cdot  e^{ig \mathcal{W}_Z} \big[ U_\tau A U_\tau^\dagger  \big] \dyad{\bs s} ) \right] + h.c. \\
    & = i \tr(  U_{t-\tau}^\dagger B U_{t-\tau} \cdot  e^{ig \mathcal{W}_Z} \big[ U_\tau A U_\tau^\dagger  \big] ) / 2^n + h.c. \\
    & = \sum_{w_Z=0}^{N} 2\sin(gw_Z) \cdot \tr( U_{t-\tau}^\dagger B U_{t-\tau} \cdot \mathcal{P}^Z_{w_Z} [ U_\tau A U_\tau^\dagger ] ) / 2^n \\
    & = \sum_{w_Z=0}^{N} 2\sin(gw_Z) \cdot R^Z_t(w_Z). \\
\end{split}
\end{equation}
aIn the second line, we use that $W_Z(\bs{s}) \ket{\bs s} = 0$, by construction, to eliminate the factor of $e^{-ig W_Z(\bs{s})}$.
In the fifth line, $\mathcal{W}_Z[P] = w_Z[P] P$ denotes the $Z$-weight superoperator.
In the sixth line, we utilize the fact, $W_Z(\bs{s})  P \ket{\bs s} = w[P] P \ket{\bs s} = \mathcal{W}_Z[ P ] \ket{\bs s}$.

From the experimental outcomes above, one can extract the reactivity for any $w_Z > 0$ via the inverse Fourier transform.
The details proceed identically to the ancilla-assisted coherent measurement protocol.
This leaves a single quantity missing: the value of the reactivity at weight $w_Z = 0$.
The reactivity at $w_Z=0$ cannot be extracted from the experiments described above, since its contribution to the experimental outcome vanishes due to the factor of $\sin(g w_Z)$.
Nonetheless, it is easily recovered by performing a single additional experiment.
Namely, one simply performs a separate measurement the value of the infinite temperature correlation function, $\tr( B U A U^\dagger )/2^n$.
This allows one to determine the weight-zero reactivity via $R_\tau^Z(0) = \tr( B U A U^\dagger )/2^n - \sum_{w_Z=1}^n R_\tau^Z(w_Z)$.

\subsection{Sampling-based measurement protocol for operator size distributions}

Operator size distributions can be measured scalably in three distinct ways.
There also exist other measurement protocols that are useful in practice but incur exponential sampling overhead when applied to large-weight operators~\cite{qi2019measuring,joshi2020quantum}.
The first two ways correspond to the coherent and decoherent protocols introduced in the previous two sections.
To measure the size distribution of a time-evolved operator $U^\dagger O U$ with these protocols, one simply sets $A = B = O$ and $U_\tau = U$, $U_{t-\tau} = U^\dagger$.
This utilizes the fact that the size distribution is equivalent to the reactivity of the Loschmidt echo circuit.

In this section, we present a third protocol to measure the operator size distribution.
This protocol exploits the fact that the size distribution is a normalized probability distribution, to perform sampling from the  distribution itself.
In general, this sampling protocol is strictly superior to the aforementioned coherent protocol, since it requires only two unitary applications and not three.
Compared to the decoherent protocol, the sampling protocol shares similar benefits and drawbacks to the coherent protocol, as discussed in the previous section.

\subsubsection{Ancilla-assisted  protocol}

Our protocol is extremely straightforward.
Following the discussion in Section~\ref{sec: coherent ancilla-assisted}, the protocol begins by preparing the EPR state between $N$ system qubits and $N$ ancilla qubits.
One then applies the unitary Pauli operation $O$ on the system side of the EPR state, and then the unitary $U^\dagger$ on the system side and $U^T$ on the ancilla side.
Equivalently, one can first apply the unitary $U$, then $O$, then $U^\dagger$, all on the system side.
This produces the state $( U O \otimes U^* ) \ket{\text{EPR}} = ( U O U^\dagger \otimes \identity ) \ket{\text{EPR}}$.
The protocol concludes by measuring this state in the Bell basis on each pair of qubits between the system and ancilla sides.

The four possible outcomes of the Bell basis measurement on the $i$-th pair of qubits are $\{ \ket{\identity_i} , \ket{X_i} , \ket{Y_i} , \ket{Z_i}\}$, where $\ket{P_i} \equiv (P_i \otimes \identity_i)\ket{\text{EPR}}_i$.
The $4^N$ possible outcomes of the $N$-fold Bell basis measurement are therefore $\{ \ket{P} = (P\otimes \identity)\ket{\text{EPR}} \}$, for all Pauli operators $P$.
Each outcome occurs with probability
\begin{equation}
    p(P) = \left| \bra{P} (U^\dagger O U \otimes \identity) \ket{\text{EPR}} \right|^2 = \frac{1}{4^n} \tr( P U^\dagger O U )^2 = |a(P; U,O)|^2,
\end{equation}
equal to the magnitude squared of the amplitude of $P$ in $U^\dagger O U$.
The probability distribution is normalized, $\sum_P p(P) = \sum_P |a(P; U,O)|^2 = \overline{\tr}(O^\dagger O) = 1$, as required.
If one associates each outcome $P$ with its weight $w[P]$, then this protocol immediately performs sampling from the size distribution, $P(w) = \sum_{P  :  w[P]=w} |a(P;U,O)|^2$.
In principle, this protocol also gives access to lots of additional information about the Pauli decomposition of $U^\dagger O U$ as well.

\subsubsection{Ancilla-free protocol}

In a similar fashion to the coherent measurement protocol, the sampling-based protocol can also be performed without any ancillas.
This allows one to sample from the modified size distribution, $P_Z(w_Z) = \sum_{P:w_Z[P]=w_Z} |a(P; U,O)|^2$, defined with respect to the $Z$-weight, $w_Z$, instead of the weight, $w$.

Our protocol is again extremely straightforward.
We prepare a random computational basis state, $\ket{\bs s}$, for $\bs s \in \{0,1\}^N$.
We then apply the unitary $U$, then $O$, then $U^\dagger$.
We conclude by measuring all $N$ qubits in the computational basis.
We repeat this procedure for many different random initial states $\ket{\bs s}$.

For a given initial state $\ket{\bs{s}}$, the probability to receive a measurement outcome $\ket{\bs s + \bs \delta}$ is equal to,
\begin{equation}
    p(\bs s + \bs \delta ; \bs s ) = \left| \bra{\bs s + \bs \delta} U^\dagger O U \ket{\bs s} \right|^2 = \left| \bra{\bs s} X_{\bs \delta} U^\dagger O U \ket{\bs s} \right|^2,
\end{equation}
where $X_{\bs \delta} = \bigotimes_{i=1}^N (X_i)^{\delta_i}$ flips each bit with $\delta_i = 1$.
We will be interested in the averaged probability distribution when $\bs{s}$ is drawn randomly from the computational basis states.
To compute this, we use the identity, $\E_{\bs s} [ \bra{\bs s} P \ket{\bs s} \bra{\bs s} Q \ket{\bs s} ] = \delta_{P=Q} \, \delta_{P \in \mathcal{Z}}$, where $\mathcal{Z}$ is the set of $Z$-basis Pauli operators, $\mathcal{Z} = \{ Z_{\bs t} \equiv \bigotimes_{i=1}^N (Z_i)^{t_i}\}$ for $\bs{t} \in \{0,1\}^N $.
This yields
\begin{equation}
    p(\bs \delta) \equiv \E_{\bs s} \left[ p(\bs s + \bs \delta ; \bs s ) \right] = \sum_P |a(P;U,O)|^2 \cdot \delta_{X_{\bs \delta} P \in \mathcal{Z}}.
\end{equation}
For each Pauli $P$, there is a unique $\bs{\delta}$ such that $X_{\bs{\delta}} P \in \mathcal{Z}$.
This $\bs \delta$ has Hamming weight $w_Z[P]$.
Thus, we can sample from the modified size distribution, $P_Z(w_Z)$, by sampling from $p(\bs \delta)$ and associating each $\bs{\delta}$ with its Hamming weight $|\bs{\delta}| = w_Z$.

\section{Quantum experiments with low reactivity can be efficiently learned}

In this section, we show that quantum experiments with low reactivity can be efficiently learned using random product state preparations and measurements~\cite{elben2022randomized}.
That is, the expectation value $\tr(O \mathcal{C}(\psi))$ for any observable $O$, quantum channel $\mathcal{C}$, and any state $\psi$ can be predicted accurately and efficiently using randomized measurements involving the channel $\mathcal{C}$, assuming the expectation value $\tr(O \mathcal{C}(\psi))$ has low reactivity.
To show this, we first summarize the standard randomized measurement protocol, and then state and prove our main result.

\vspace{4mm}
\noindent \textbf{\emph{Background on the randomized measurement toolbox.}} We are interested in characterizing the properties of a quantum channel $\mathcal{C}$ on $N$ qubits.
We assume that one has access to $\mathcal{C}$ through quantum experiments of the form: prepare a random product state, apply $\mathcal{C}$, measure in a random product basis. 
We would like to know what features of the channel can be learned from this extremely simple form of access.
In particular, we will ask: for what observables $O$ and states $\psi$ can the expectation values $\tr(O \mathcal{C}(\psi))$ be efficiently estimated from randomized measurement data?

We assume that the random product basis for both the state preparation and measurement is a tensor product of single-qubit 2-designs on each qubit.
To be specific, we focus on the single-qubit stabilizer states, $\mathcal{S} = \{ \ket{0}, \ket{1}, \ket{+}, \ket{-}, \ket{+i}, \ket{-i} \}$, which are a 2-design.
We let $\ket{u} = \otimes_{i=1}^n \ket{u_i}$ be a stabilizer product state, where each $\ket{u_i} \in \mathcal{S}$.
In each round $r$ of the experiment, we (1) prepare a random product state $\ket{u_r}$ (with probability $(1/6)^N$ for each $u_r$), (2) apply $\mathcal{C}$, and (3) receive a random product measurement outcome $\ket{v_r}$ (with probability $(1/3)^N \bra{v_r} \mathcal{C}(\dyad{u_r}) \ket{v_r}$ for each $v_r$).
We repeat this for a total of $M$ rounds.

Our result will rely on only one simple property of the random measurement toolbox: that it allows one to estimate Pauli transition amplitudes, $e_{PQ} \equiv \frac{1}{2^n} \tr(P \mathcal{C}(Q))$, to within mean-square error $3^{w[P]+w[Q]}/M$, where $M$ is the number of random measurement rounds.
Our estimate $\hat{e}_{PQ}$ is formed as follows,
\begin{equation}
    \hat{e}_{PQ} = 3^{w[P]+w[Q]} \cdot \frac{1}{M}  \sum_{r=1}^M \bra{v_r} P \ket{v_r} \bra{u_r} Q \ket{u_r}.
\end{equation}
Using 2-design properties, one can easily show that $\hat{e}_{PQ}$ provides an unbiased estimate of $e_{PQ}$, with a mean-square error $\text{Var}(\hat{e}_{PQ}) \leq 3^{w[P]+w[Q]}/M$.
Intuitively, the factors of $3^{w[P]}$ and $3^{w[Q]}$ correspond to the (inverse) probabilities that the measurement basis $v_r$ commutes with $P$, and $u_r$ with $Q$, respectively.

\vspace{4mm}
\noindent \textbf{\emph{Proof of efficiency.}} We are now in a position to state our result, that quantum experiments can be efficiently learned from randomized measurement data if they have low reactivity. This result follows extremely simply from the definition of the reactivity, and the standard randomized measurement analysis described above.

\begin{theorem}[Quantum experiments with low reactivity can be efficiently learned]
	Consider any local observable $O$, quantum channel $\mathcal{C}$, and state $\psi$ on $N$ qubits.
    Suppose that the cumulative reactivity, $H(w_*) = \sum_{w > w_*} R_{\tau=0}(w)$, where $R_{\tau=0}(w) = \tr(O \mathcal{C}( \mathcal{P}_w[\psi]))$, of the expectation value is bounded above by $|H(w_*)| \leq \varepsilon$ for some positive integer $w_*$.
    Then the expectation value $\tr(O \mathcal{C}(\psi))$ can be estimated using randomized product state preparations and measurements on $\mathcal{C}$ to within mean-square error  $\varepsilon^2 + n^{\mathcal{O}(w_*)}/M$.	
\end{theorem}
\noindent Our proof is exceedingly short and is provided below. 
We suspect that the mean-square error in the theorem can be improved to $\varepsilon^2 + 2^{\mathcal{O}(w_*)}/M$ rather than $\varepsilon^2 + n^{\mathcal{O}(w_*)}/M$ using the techniques of Ref.~\cite{huang2022learning}. 
However, we do not do so here, as this is not the main focus of our work.

\begin{proof}
	Let $O = \sum_P o_P P$ denote the Pauli decomposition of $O$, where $o_P = \overline{\text{tr}}(O P)$ for $\overline{\text{tr}}(\cdot) \equiv \frac{1}{2^n} \tr(\cdot)$ and $w[P] \leq k = \mathcal{O}(1)$ for every $P$ with $o_P$ nonzero by assumption.
    Let $\psi = \frac{1}{2^n} \sum_Q c_Q Q$ denote the Pauli decomposition of $\psi$, where $c_Q = \tr(Q \psi)$ is the expectation value of $Q$ in the state $\psi$.
    Our expectation value of interest can be written as
        \begin{equation}
    \begin{split}
        \tr(O\mathcal{C}(\psi))\
        & = \sum_{w\geq 0} R_{\tau=0}(w) \\ & = H(w_*) +   \sum_{Q:w[Q]\leq w_*} c_Q \cdot \overline{\text{tr}}( O \mathcal{C}( Q ))  \\
        & =H(w_*) + \sum_{P:w[P]\leq k} \sum_{Q:w[Q]\leq w_*} o_P c_Q \cdot \overline{\text{tr}}( P \mathcal{C}( Q )),
    \end{split}
    \end{equation}
    where in the second line of the right side, we split the sum into  terms with weights less than or equal to $w_*$ and more than $w_*$ (i.e.\ $H(w_*)$), insert the  Pauli decomposition of $\psi$ and use     $R(w)  = \sum_{Q:\, w[Q]=w} c_Q \, \overline{\tr}(O \mathcal{C}(Q))$. In the third line, we insert Pauli decomposition of the local observable $O$. We note that $|H(w_*)| \leq \varepsilon$ by assumption.

    Now, we can use randomized product state preparations and measurements to compute an unbiased estimate $\hat{e}_{PQ}$ each inner product $e_{PQ} = \overline{\text{tr}}(P \mathcal{C}(Q))$, to within mean-square error $\text{Var}(\hat{e}_{PQ}) = 3^{w[P]+w[Q]}/M$ using $M$ random product states.
    This implies that we can compute an unbiased estimate of the low-weight component of our expectation value of interest, to within mean-square error
    \begin{equation}
    \begin{split}
        \text{Var}\left( \sum_{P:w[P]\leq k} \sum_{Q:w[Q]\leq w_*} o_P c_Q \hat{e}_{PQ} \right)
        & \leq \left( \sum_{P:w[P]\leq k} \sum_{Q:w[Q]\leq w_*} |o_P c_Q| \cdot \sqrt{\text{Var}(\hat{e}_{PQ})} \right)^2 \\
        & \leq \frac{1}{M} \cdot 3^{k+w_*} \left( \sum_{P:w[P]\leq k} \sum_{Q:w[Q]\leq w_*} |o_P c_Q| \right)^2  \\
        & \leq \frac{1}{M} \cdot 3^{k+w_*} \cdot {n \choose k} 4^{k} \cdot {n \choose w_*} 4^{w_*} \cdot \sum_{P:w[P]\leq k} |o_P|^2 \cdot \sum_{Q:w[Q]\leq w_*} |c_Q|^2  \\
        & \leq \frac{1}{M} \cdot 3^{k+w_*} \cdot {n \choose k} 4^{k} \cdot {n \choose w_*} 4^{w_*} \cdot 1 \cdot {n \choose w_*} 4^{w_*}  \\
        & \leq \frac{1}{M} (12n)^k (48n)^{w_*} = n^{\mathcal{O}(w_*)}/M.  \\
    \end{split}
    \end{equation}
    Here, in the first inequality, we upper bound $\text{Cov}(A,B) \leq \sqrt{\text{Var}(A)\text{Var}(B)}$; in the second inequality, we upper bound $\text{Var}(\hat{e}_{PQ}) \leq 3^{k+w_*}/M$; in the third inequality, we apply $\left( \sum_x |x| \right)^2 \leq \sum_x 1 \cdot \sum_x |x|^2$, and upper bound the number of Paulis with weight less than $k$ as ${n \choose k} 4^k$, and similar for $w_*$; in the fourth inequality, we use $\sum_{P} |o_P|^2 \leq \overline{\text{tr}}(O^2) \leq 1$ and $|c_Q|^2 \leq 1$; and in the fifth inequality, we combine terms.
    Finally, we note that the mean-square error of our estimate compared to the exact expectation value (including high-weight components) is equal to $\varepsilon^2 + n^{\mathcal{O}(1)}/M$, since the mean-square error of the randomized measurements and the systematic bias due to $|H(w_*)|\leq \varepsilon$ add in quadrature.
\end{proof}

\section{Quantum experiments with low reactivity can be efficiently fast-forwarded}

In this section, we show that quantum experiments with low reactivity can be efficiently fast-forwarded using random product state preparations and measurements~\cite{elben2022randomized}.
That is, the expectation value $\tr(O \mathcal{C}(\psi))$ for any observable $O$, quantum channel $\mathcal{C} = \mathcal{C}_2 \circ \mathcal{C}_1$, and any state $\psi$ can be predicted accurately and efficiently using randomized measurements involving the  channels $\mathcal{C}_1$ and $\mathcal{C}_2$ individually, assuming the expectation value $\tr(O \mathcal{C}(\psi))$ has low reactivity.
To show this, we first introduce a simple randomized measurement protocol for fast-forwarding quantum dynamics, and then prove that our protocol succeeds for quantum experiments with low reactivity.

\vspace{4mm}
\noindent \textbf{\emph{A randomized measurement protocol for fast-forwarding quantum dynamics.}} We consider the following randomized measurement protocol (see the previous section for an introduction to notation).
We consider two randomized measurement experiments.
In each round $r$ of the first experiment, we (1) prepare the state $\psi$, (2) apply the first channel $\mathcal{C}_1$, and (3) measure in a random product basis, receiving an outcome $\ket{v_r}$ (with probability $(1/3)^N \bra{v_r} \mathcal{C}_1(\psi) \ket{v_r}$ for each $v_r$).
In each round $r$ of the second experiment, we (1) prepare a random product state $\ket{u_r}$ (with probability $(1/6)^N$ for each $u_r$), (2) apply the second channel $\mathcal{C}_2$, and (3) measure $O$ with outcome $o_r$.
We repeat this for a total of $M$ rounds for each experiment.

To see how these two experiments can be used to compute the expectation value $\tr(O \mathcal{C}(\psi))$, let us expand the expectation value as follows,
\begin{equation} \label{eq: bP aP decomposition}
    \tr(O \mathcal{C}(\psi)) = \text{tr}(O \mathcal{C}_2(\mathcal{C}_1(\psi))) = \sum_P \overline{\text{tr}}( O \mathcal{C}_2(P)) \tr(P \mathcal{C}_1(\psi)) = \sum_P b_P a_P,
\end{equation}
where we abbreviate $a_P \equiv \tr(P \mathcal{C}_1(\psi))$ and $b_P \equiv \overline{\text{tr}}( O \mathcal{C}_2(P))$.
The aim of our fast-forwarding protocol will be to estimate the coefficients $a_P$ and $b_P$ from the first and second randomized measurement experiment above, respectively.
From these, the total expectation value can be computed directly from the equation above.

We compute our estimate $\hat{a}_P$ of the coefficient $a_P$ as follows,
\begin{equation}
    \hat{a}_P = 3^{w[P]} \cdot \frac{1}{M} \sum_{r=1}^M \bra{v_r} P \ket{v_r}.
\end{equation}
Here, $v_r$ are the product states obtained from the $M$ measurement outcomes in the experiment, and hence occur with probability $\bra{v_r} \mathcal{C}_1(\psi) \ket{v_r}$.
We compute our estimate $\hat{b}_P$ of the coefficient $b_P$ as follows,
\begin{equation}
    \hat{b}_P = 3^{w[P]} \cdot \frac{1}{M} \sum_{r=1}^M \bra{u_r} P \ket{u_r} o_r.
\end{equation}
One can easily verify that both estimates converge to their desired values in expectation, with a mean-square error $\text{Var}(\hat{a}_P), \text{Var}(\hat{b}_P) \leq 3^{w[P]}/M$.

\vspace{4mm}
\noindent \textbf{\emph{Proof of efficiency.}}
A straightforward analysis shows that the above randomized measurement fast-forwarding protocol is efficient whenever the quantum experiment has low reactivity.
Here, the reactivity is measured at the same time as the Pauli decomposition, i.e.~in between the applications of $\mathcal{C}_1$ and $\mathcal{C}_2$.
This is summarized in the following theorem.

\begin{theorem}[Quantum experiments with low reactivity can be efficiently fast-forwarded]
	Consider any local observable $O$, quantum channel $\mathcal{C} = \mathcal{C}_2 \circ \mathcal{C}_1$, and state $\psi$ on $N$ qubits.
    Suppose that the cumulative reactivity, $H(w_*) = \sum_{w > w_*} R(w)$, where $R(w) = \tr(O \mathcal{C}_2( \mathcal{P}_w[ \mathcal{C}_1(\psi)]))$, of the expectation value is bounded above by $|H(w_*)| \leq \varepsilon$ for some positive integer $w_*$.
    Then the expectation value $\tr(O \mathcal{C}(\psi))$ can be estimated using randomized product state preparations and measurements on $\mathcal{C}_1$ and $\mathcal{C}_2$ individually, to within mean-square error  $\varepsilon^2 + n^{\mathcal{O}(w_*)}/M$.	
\end{theorem}
\noindent As in the previous section on efficient learning, our proof is exceedingly short and is provided below. 

\begin{proof}
    From Eq.~(\ref{eq: bP aP decomposition}) and the assumptions of the theorem, we can write,
    \begin{equation}
        \tr(O \mathcal{C}(\psi)) = \sum_P b_P a_P = H(w_*) + \sum_{P : w[P] \leq w_*} b_P a_P
    \end{equation}
	where $a_P = \tr(P \mathcal{C}_1(\psi))$ and $b_P = \overline{\text{tr}}(O \mathcal{C}_2(P))$ are defined as before.
    As discussed, each coefficient $b_P$ and $a_P$ can be estimated to within mean-square error $3^{w[P]} \leq 3^{w_*}/M$.
    Since these estimates arise from different randomized measurements and hence are uncorrelated with one another, this implies that one can form an unbiased estimate of the product $b_P a_P$ with mean-square error less than $(9^{w_*} + 3^{w_*}(|a_P|^2 + |b_P|^2))/M = 2^{\mathcal{O}(w_*)}/M$.
    Following identical steps to our proof of efficient learning, this yields a mean-square error in our final estimate of the expectation value that is less than $\varepsilon^2 + n^{\mathcal{O}(w_*)}/M$.
\end{proof}

\section{Rigorous classical algorithm for random quantum circuits with low reactivity}

In this section, we provide a rigorous proof that random quantum circuits with low reactivity can be efficiently classically simulated via local information algorithms.
This proof is intended as a short example, to illustrate what aspects of classical simulability and the reactivity are able to be addressed by existing rigorous techniques.
This complements the extensive numerical studies presented in the main text, which support the use of the reactivity as a proxy for the complexity of local information algorithms in a much broader class of physical settings.

Our first result considers global reactivity functions, defined over the entire circuit volume.
It is exceedingly simple, and applies to any quantum circuit.
\begin{theorem}
    [Quantum experiments with low global reactivity can be efficiently classically simulated]
    Consider any local observable $O$, quantum circuit $\mathcal{C}$, and state $\psi$ on $N$ qubits.
    Suppose that the cumulative global reactivity, $C_g(w_*) = \sum_{w > w_*} R_g(w)$, where $R_g(w) = \sum_{\vec{P}} A_{\vec{P}} \delta_{w[\vec{P}]=w}$, is bounded above by $|C_g(w_*)| \leq \varepsilon$. 
    Then the expectation value $\tr(O \mathcal{C}(\psi))$ can be computed to precision $\varepsilon$ in time $2^{\mathcal{O}(w_*)}$ by a classical local information algorithm that truncates Pauli paths with global weight greater than $w_*$. 
\end{theorem} \label{thm:4}
\noindent 

Our second result considers individual reactivity functions defined on each individual time slice, and applies to random quantum circuits (defined below).
\begin{theorem}
    [Quantum experiments with low reactivity can be efficiently classically simulated]
    Consider any local observable $O$, Pauli-random quantum circuit $\mathcal{C}$, and state $\psi$ on $N$ qubits.
    Suppose that the summed reactivities, $\sum_{\tau=0}^t C_\tau(w_*) = \sum_{\tau=0}^t \sum_{w > w_*} R_\tau(w)$, where $R_\tau(w) = \sum_{\vec{P}} A_{\vec{P}} \delta_{w[P_t]=w}$, are bounded above by $\sum_{\tau=0}^t  \E_{\mathcal{C}} |C_{\tau,\mathcal{C}}(w_*)|^2 \leq \varepsilon^2$. 
    Then the expectation value $\tr(O \mathcal{C}(\psi))$ can be computed within mean-square error $\varepsilon$ in time $N^{\mathcal{O}(w_*)}$ by a classical local information algorithm that truncates Pauli operators with weight greater than $w_*$ at each time step. 
\end{theorem} \label{thm:5}

\vspace{4mm}
\noindent \textbf{\emph{Setup.}}
We consider quantum circuits in which random single-qubit Pauli rotations are inserted in between each fixed two-qubit circuit gate.
That is, the circuit is given by $U = u_t U_t \ldots  u_2 U_2 u_1 U_1 u_0$, where each $U_\tau$ is an arbitrary layer of any two-qubit circuit gates, and each $u_\tau = \otimes_{i=1}^n u_{\tau,i}$ is a tensor product of random single-qubit Pauli operations, $u_{\tau,i} \sim \{\identity_i,X_i,Y_i,Z_i\}$.
No classical algorithm for such circuits is known beyond exact simulation (with computational cost $2^n$).
We provide a short comparison between our result and previous work on classically simulating more specific classes of random circuits~\cite{aharonov2023polynomial,angrisani2024classically} after introducing our result.

The Pauli path decomposition of the experiment can be written as,
\begin{equation}
    \tr(O U \psi U^\dagger) = \sum_{\vec{P}} (-1)^{\vec{P} \cdot \vec{u}} A_{\vec{P}},
\end{equation}
where $A_{\vec{P}}$ are the path amplitudes for the bare experiment in the absence of any random Pauli rotations, and the term $(-1)^{\vec{P} \cdot \vec{u}}$ captures the phase applied to each Pauli path by the rotation.
Here, in an abuse of notation, $\vec{P}$ denotes the vector of length $2(t+1)N$ with entries $P_{1,\tau,i} = 1$ if $P_{\tau,i}$ is $X$ or $Y$, and $P_{1,\tau,i} = 0$ otherwise, and $P_{2,\tau,i} = 1$ if $P_{\tau,i}$ is $Z$ or $Y$, and $P_{2,\tau,i} = 0$ otherwise.
We have the same definition for $\vec{u}$, with the first index swapped between $1$ and $2$.
Hence, the inner product $\vec{P} \cdot \vec{u}$ encodes whether the Pauli path $\vec{P}$ commutes or anti-commutes the random Pauli rotations $\vec{u} = (u_0,\ldots,u_t)$.
When averaged over the random Pauli rotations, we have the convenient property
\begin{equation}
    \E_{\vec{u}} \left[ (-1)^{\vec{P} \cdot \vec{u}} (-1)^{\vec{Q} \cdot \vec{u}} \right]
    =
    \delta_{\vec{P},\vec{Q}}
\end{equation}
for all $\vec{P}$, $\vec{Q}$.

\vspace{4mm}
\noindent \textbf{\emph{Proofs of efficiency.}}

\begin{proof}[Proof of Theorem~3]
We can decompose the expectation value into contributions from Pauli paths of weight greater than $w_*$, which sum to $C_g(w_*)$, and contributions from paths of weight less than or equal to $w_*$,
\begin{equation}
    \tr(O U \psi U^\dagger) = C_g(w_*) + \!\!  \sum_{\vec{P}:w[\vec{P}]\leq w_*} \!\! A_{\vec{P}}.
\end{equation}
From Lemma~8 of Ref.~\cite{aharonov2023polynomial}, there are at most $2^{\mathcal{O}(w_*)}$ Pauli paths of weight greater than $w_*$ if $O$ is local (see also Appendix~C of Ref.~\cite{schuster2025polynomial}).
The amplitude of each Pauli path can be computed in $\mathcal{O}(w_*)$ time.
Hence, the expectation value can be computed to within error $|C_g(w_*)|$ in $2^{\mathcal{O}(w_*)}$ time.
\end{proof}

\begin{proof}[Proof of Theorem~4]
We can decompose the experiment into contributions from Pauli paths that are truncated (left) versus kept (right) by the classical algorithm,
\begin{equation}
    \tr(O U \psi U^\dagger) = 
    \!\!\!\!\!\!\!\!\!\!  
    \sum_{\substack{\text{paths $\vec{P}$ such that} \\ \text{$w[P_\tau] > w_*$ for some $\tau$}}} 
    \!\!\!\!\!\!\!\!\!\!\!\!\!\!\!   
    (-1)^{\vec{P} \cdot \vec{u}} A_{\vec{P}}
    \,\,\,\,\, + 
    \!\!\!\!\!\!\!\!\!\!  
    \sum_{\substack{\text{paths $\vec{P}$ such that} \\ \text{$w[P_\tau] \leq w_*$ for all $\tau$}}}
    \!\!\!\!\!\!\!\!\!\!\!  
    (-1)^{\vec{P} \cdot \vec{u}} A_{\vec{P}}.
\end{equation}
The mean-square error of the classical algorithm is equal to the mean-square of the left summation.
Since the circuit is Pauli-random, the mean-square sum over Pauli path contributions is equal to the sum of squared contributions,
\begin{equation}
    \E_{\vec{u}} \Bigg[ \Bigg( \sum_{\substack{\text{paths $\vec{P}$ such that} \\ \text{$w[P_\tau] > w_*$ for some $\tau$}}} 
    \!\!\!\!\!\!\!\!\!\!\!\!\!\!\!   
    (-1)^{\vec{P} \cdot \vec{u}} A_{\vec{P}} \Bigg)^2 \Bigg]
    \,\,\,\,\,\, =
    \!\!\!\!\!\!\!\!\!\!  
    \sum_{\substack{\text{paths $\vec{P}$ such that} \\ \text{$w[P_\tau] > w_*$ for some $\tau$}}} \!\!\!\!\!\!\!\!\!\!  
    | A_{\vec{P}} |^2.
\end{equation}
As the latter sum is over positive quantities, it can be upper bounded as follows,
\begin{equation}  
    \sum_{\substack{\text{paths $\vec{P}$ such that} \\ \text{$w[P_\tau] > w_*$ for some $\tau$}}} \!\!\!\!\!\!\!\!\!\!  
    | A_{\vec{P}} |^2
    \,\,\,\,\,\, \leq
    \,\,\,\,\,\, \sum_{\tau=0}^t
    \Bigg( \sum_{\substack{\text{paths $\vec{P}$ such that} \\ \text{$w[P_\tau] > w_*$}}} | A_{\vec{P}} |^2 \Bigg).
\end{equation}
On the left hand side, each truncated Pauli path is counted once; on the right hand side, each truncated Pauli path is counted a number of times equal to the number of time slices where its weight is greater than $w_*$.
Finally, each sum over squared Pauli path contributions can be re-written in terms of the mean-square cumulative reactivity at that time slice, using the same identity as before,
\begin{equation}
    \sum_{\substack{\text{paths $\vec{P}$ such that} \\ \text{$w[P_\tau] > w_*$}}} \!\!\!\!\!\!\!\!\! | A_{\vec{P}} |^2
    \,\, = \,\,
    \E_{\vec{u}} 
    \Bigg[
    \Bigg( \sum_{\substack{\text{paths $\vec{P}$ such that} \\ \text{$w[P_\tau] > w_*$}}} \!\!\!\!\! (-1)^{\vec{P} \cdot \vec{u}} A_{\vec{P}} \Bigg)^2 \Bigg] 
    = \E_{\mathcal{C}} \left[ C_{\tau,\mathcal{C}}(w_*)^2 \right],
\end{equation}
where we note that $\E_{\vec{u}} = \E_{\mathcal{C}}$.
The mean-square error of the classical algorithm is less than $\sum_{\tau=0}^t \E_{\mathcal{C}} \left[ C_{\tau,\mathcal{C}}(w_*)^2 \right]$, as desired.
There are at most $N^{\mathcal{O}(w_*)}$ Pauli operators of weight $w_*$, and hence our algorithm takes at most $N^{\mathcal{O}(w_*)}$ time.
\end{proof}

\vspace{4mm}
\noindent \textbf{\emph{Comparison to previous works~\cite{aharonov2023polynomial,angrisani2024classically}.}} Our Theorems~\ref{thm:4} and~\ref{thm:5} generalize classical simulation algorithms introduced in two lines of recent work.
First, when quantum circuits with random single-qubit Pauli rotations are \emph{noisy}, they can be classically simulated with a local information algorithm in time $n^{\mathcal{O}(1/\gamma)}$, where $\gamma$ is the noise rate~\cite{aharonov2023polynomial}. 
This can be viewed as a specific case of our statement in Theorem~\ref{thm:4}.
Noise damps the contribution of high-weight components, leading to a mean-square cumulative reactivity $\E_U C_g(w_*)^2 = \sum_{\vec{P}:w[\vec{P}]>w_*} e^{-\gamma w[\vec{P}]} A_{\vec{P}}^2 \leq e^{-\gamma w[\vec{P}]}$.
Second, when such circuits have random single-qubit \emph{Clifford} rotations instead of random single-qubit Pauli rotations, they can be classically simulated with a local information algorithm even in the absence of noise~\cite{angrisani2024classically}. 
This can be viewed as a specific case of our statement in Theorem~\ref{thm:5}.
The mean-square reactivity under random single-qubit Clifford rotations is given by, $\E_U R(w;U)^2 = \E_u \text{tr}( U_2^\dagger O U_2 \cdot u \cdot \mathcal{P}_w[ U_1 \psi U_1^\dagger] \cdot u^\dagger)^2 = 3^{-w} \sum_{P:w[P]=w} \text{tr}( U_2^\dagger O U_2 P)^2 \text{tr}(P U_1 \psi U_1^\dagger)^2 \leq 3^{-w}$.
This decays exponentially in the weight $w$.
Setting a desired precision $\varepsilon$ and $w_* = \mathcal{O}(\log 1/\varepsilon)$ yields an efficient classical algorithm with runtime $n^{\mathcal{O}(1/\varepsilon)}$~\cite{angrisani2024classically}.

\end{document}